%% file: main.tex
\PassOptionsToPackage{svgnames}{xcolor}
\documentclass[10pt, conference]{IEEEtran}
\usepackage[keeplastbox]{flushend}

\usepackage{preamble}
\usepackage{macros}
\usepackage{soul}

\title{Data-Driven Synthesis of Provably Sound Side~Channel Analyses}

\author{
  \IEEEauthorblockN{
    Jingbo Wang,
    Chungha Sung,
    Mukund Raghothaman and 
    Chao Wang
  }
  \IEEEauthorblockA{
    University of Southern California 
  }
}

\begin{document}
\maketitle

\begin{abstract}
We propose a \emph{data-driven} method for synthesizing static
analyses to detect side-channel information leaks in cryptographic
software.  Compared to the conventional way of manually crafting such
static analyzers, which can be tedious, error prone and suboptimal,
our \emph{learning-based} technique is not only automated but also
provably sound.
Our analyzer consists of a set of type-inference rules learned from
the training data, i.e., example code snippets annotated with the ground
truth.  Internally, we use \emph{syntax-guided synthesis (SyGuS)} to
generate new recursive features and \emph{decision tree learning (DTL)} to
generate analysis rules based on these features.  We guarantee
soundness by proving each learned analysis rule via a technique called
\emph{query containment checking}.
We have implemented our technique in the LLVM compiler and used it to
detect \emph{power side channels} in C programs that implement
cryptographic protocols.  Our results show that, in addition to being
automated and provably sound during synthesis, our analyzer
can achieve the same empirical accuracy as two state-of-the-art,
manually-crafted analyzers while being 300X and 900X faster,
respectively.
\end{abstract}

\section{Introduction}
\label{sec:introduction}

Static analyses are being increasingly used to detect security
vulnerabilities such as \emph{side channels}~\cite{wang2019mitigating,
zhang2018sc,chen2017precise,brennan2018symbolic}.  However, manually
crafting static analyzers to balance between accuracy and efficiency
is not an easy task: even for domain experts, it can be labor
intensive, error prone, and result in suboptimal implementations.
For example, we may be tempted to add expensive analysis rules for
specific sanitized patterns without realizing they are rare in target
programs.  Even if the analysis rules are carefully tuned to a corpus
of code initially, they are unresponsive to changing characteristics
of the target programs and thus may become suboptimal over time;
manually updating them to keep up with new programs would be
difficult.

One way to make better accuracy-efficiency trade-offs and to
dynamically respond to the distribution of target programs is to use
data-driven approaches~\cite{bielik2017learning,mendis2019compiler}
that automatically synthesize analyses from labeled examples provided
by the user.  However, checking soundness or compliance with user
intent (generalization) has always formed a significant challenge for
example-based synthesis techniques~\cite{gulwani2012spreadsheet,
  leung2015interactive, singh2012synthesizing, smith2016mapreduce,
  solar2006combinatorial}.  The lack of soundness guarantees, in
particular, hinders the application of such learned analyzers in
security-critical applications. While several existing
works~\cite{an2019augmented, mayer2015user, singh2015predicting,
  devlin2017robustfill} try to address this problem, rigorous
soundness guarantees have remained elusive.

To overcome this problem, we propose a learning-based method for 
synthesizing a
\emph{provably-sound} static analyzer that detects side channels
in cryptographic software, by inferring a
\emph{distribution} type for each program variable that indicates
if its value is statistically dependent on the secret.
The overall flow of our method, named \GPS{}, is shown in
Fig.~\ref{overview}. The input is a set of training data and the
output is a learned analyzer.  The training data are small programs
annotated with the ground truth, e.g., which program variables have
leaks.

\setlength{\textfloatsep}{10pt plus 1.0pt minus 2.0pt}
\begin{figure}
\centering
\scalebox{0.85}{\input{images/introduction/jigsaw.tex}}
\vspace{-.5em}
\caption{The overall flow of \GPS{}, our data-driven synthesis method.}
\label{overview}
\end{figure}
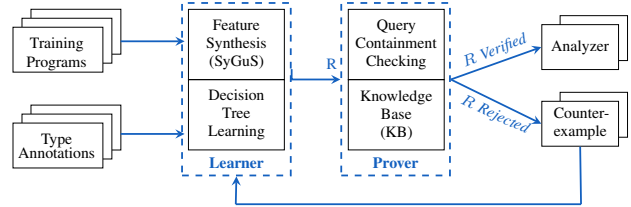

Internally, \GPS{} consists of a \emph{learner} and a \emph{prover}.
The \emph{learner} uses syntax guided synthesis (SyGuS) to generate
recursive features and decision tree learning (DTL) to generate
type-inference rules based on these features; it returns a set $R$ of
Datalog formulas that codify these rules.  The \emph{prover} checks
the soundness of each learned rule, i.e., it is not only consistent
with the training examples but also valid for any unseen
programs. This is formulated by solving a \emph{query containment
checking} problem, i.e., each rule must be justified by existing proof
rules called the knowledge base ($KB$).  Since only proved rules are
added to the analyzer, the analyzer is guaranteed to be sound. If a
rule cannot be proved, we add a counter-example to prevent
the \emph{learner} from generating it again.

We have implemented \GPS{} in LLVM and evaluated it on 568 C programs
that implement cryptographic protocols and
algorithms~\cite{barthe2015verified,blomer2004provably,bayrak2013sleuth}.
Together, they have 2,691K lines of C code.  We compared our learned
analyzer with two state-of-the-art, hand-crafted side-channel analysis
tools~\cite{zhang2018sc,wang2019mitigating}.
Our experiments show that the learned analyzer achieves the same
empirical accuracy as the two state-of-the-art tools, while being
several orders-of-magnitude faster.  Specifically, \GPS{} is, on
average, 300$\times$ faster than the analyzer from
\cite{wang2019mitigating} and 900$\times$ faster than the analyzer
from \cite{zhang2018sc}.

To summarize, this paper makes the following contributions:
\begin{itemize}
\item
We propose the first data-driven method for learning a provably sound
static analyzer using syntax guided synthesis (SyGuS) and decision
tree learning (DTL).
\item
We guarantee soundness by formulating and solving a Datalog query
containment checking problem.
\item
We demonstrate the effectiveness of our method for
detecting side channels in cryptographic software.
\end{itemize}


In the remainder of this paper, we begin by presenting the 
technical background in Section~\ref{sec:prelims} and our motivating
example in Section~\ref{sec:motivation}.
We then describe the \emph{learner} in Section~\ref{sec:learner} and
the \emph{prover} in Section~\ref{sec:prover}, followed by the
experimental results in Section~\ref{sec:experiment}.  Finally, we
survey the related work in Section~\ref{sec:related} and conclude in
Section~\ref{sec:conclusion}.

\section{Preliminaries}
\label{sec:prelims}


\subsection{Power Side-Channels}

Prior works in side-channel
security~\cite{ishai2003private,balasch2014cost,coron2012conversion}
show that variance in the power consumption of a computing device may
leak secret information; for example, when a secret value is stored in
a physical register, its number of logical-1 bits may affect the power
consumption of the CPU.  Such side-channel leaks are typically
mitigated by \emph{masking}, e.g., using $d$ random bits
($r_1,\ldots,r_{d}$) to split a $key$ bit into $d + 1$ secret shares:
$key_1 = r_1$, $\ldots$, $key_d = r_d$, and $key_{d+1} = r_1 \oplus
r_2 \ldots \oplus r_d \oplus key$, where $\oplus$ denotes the logical
operation \emph{exclusive or} (XOR).  Since all $d+1$ shares are
uniformly distributed in the $\{0,1\}$, in theory, this
\emph{order-d masking} scheme is secure in that any combination of
less than $d$ shares cannot reveal the secret, but combining all $d+1$
shares, $key_1 \oplus key_2 \oplus ... key_{d+1}$ = $key$, recovers
the secret.

In practice, masking countermeasures must also be implemented properly to avoid
de-randomizing any of the secret shares accidentally.  Consider $t$ =
$t_L \oplus t_R$ = ($r_1 \oplus key$) $\oplus$ ($r_1 \oplus b$) = $key
\oplus b$.  While syntactically dependent on the two randomized values
$t_L$ and $t_R$, $t$ is in fact leaky because, semantically, it does
not depend on the random input $r_1$.
In this work, we aim to learn a static analyzer that can soundly prove
that \emph{all intermediate variables of a program} that implements 
masking countermeasures are free of such leaks.

\subsection{Type Systems}
\label{typeSys}

Type systems prove to be effective in analyzing power side
channels~\cite{zhang2018sc,wang2019mitigating}, e.g., by certifying
that all intermediate variables of a program are \emph{statistically
independent} of the secret.
Typically, the program inputs are marked as public (\texttt{INPUB}),
secret (\texttt{INKEY}) or random (\texttt{INRAND}), and then the
types of all other program variables are inferred automatically.

The type of a variable $v$, denoted \texttt{TYPE(v)}, may be
\texttt{RUD}, \texttt{SID}, or \texttt{UKD}.
Here, \texttt{RUD} stands for random uniform distribution, meaning $v$
is either a random bit or being masked by a random bit.
\texttt{SID} stands for secret independent distribution, meaning 
$v$ does not depend on the secret.
While an \texttt{RUD} variable is, by definition, also \texttt{SID},
an \texttt{SID} variable does not have to be \texttt{RUD} (e.g.,
variables that are syntactically independent of the secret).
Finally, \texttt{UKD} stands for unknown distribution, or potentially
leaky; if the analyzer cannot prove $v$ to be \texttt{RUD}
or \texttt{SID}, then it is assumed to be \texttt{UKD}.

Type systems are generally designed to be sound but not necessarily
complete.
They are sound in that they never miss real leaks. 
For example, by default, they may safely assume that all variables are
\texttt{UKD}, unless a variable is specifically elevated to $\SID$ or $\RUD$ 
by an analysis rule.
Similarly, they may conservatively classify \texttt{SID} variables as
\texttt{UKD}, or classify \texttt{RUD} variables as \texttt{SID},
without missing real leaks. 
In general, the sets of variables that can be marked as the three
types form a hierarchy: $S_{\texttt{RUD}}$ $\subseteq$
$S_{\texttt{SID}}$ $\subseteq$ $S_{\texttt{UKD}}$.

\subsection{Relations}

A program in static single assignment (SSA) format can be represented 
as an abstract syntax tree (AST).
Static analyzers infer the type of 
each node $x$ of the program's AST
based on various \emph{features} of $x$.
In this work, pre-defined features are represented as
\emph{relations}.
\begin{itemize}
\item 
Unary relations $\INPUBOP(x)$, $\INKEYOP(x)$, and $\INRANDOP(x)$
denote the given security level of a program input $x$, which may be
public, secret, or random.
\item 
Unary relations $\RUD(x)$, $\SID(x)$, and $\INRANDOP(x)$ denote the
inferred type of a program variable $x$, which may be uniformly
random, secret independent, or unknown.
\item 
Unary relation $\OPOP(x)$ denotes the operator type of the AST node
$x$, e.g., $\OPOP(x)$ := $\ANDOROP(x)$~$|$~$\XOROP(x)$, where
$\ANDOROP(x)$ means that $x$'s operator type is either \emph{logical
  and} or \emph{logical or}, and $\XOROP(x)$ means that $x$'s operator
type is \emph{exclusive or};
\item 
Binary relations $\LCOP(x,L)$ and $\RCOP(x, R)$ indicate that the AST
nodes $L$ and $R$ are the left and right operands of $x$,
respectively.
\item 
Binary relation $\SUPPOP(x,y)$ indicates that the AST node $y$ is used
in the computation of $x$ syntactically, while $\DOMOP(x,y)$ indicates
that random program input $y$ is used in the computation of $x$
semantically.
\end{itemize} 
%
%

\section{Motivation}
\label{sec:motivation}

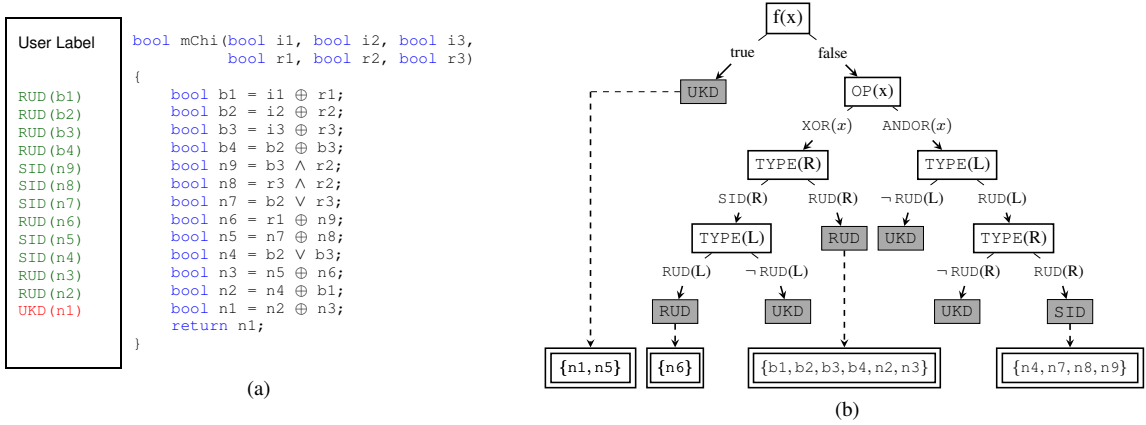
\begin{figure*} [t!] 
  \centering
  \subfloat[]
  {
  \hspace{-1em}
    \scalebox{0.75}{
      \renewcommand{\arraystretch}{1.3}
      \begin{tabular}{|p{0.1\textwidth}@{}|p{0.35\textwidth}}
        \cline{1-1}
        %
        \begin{minipage}{0.16\textwidth}
          \textcolor{DarkGreen}{
\begin{lstlisting}
!\textcolor{black}{\textsf{User Label}}!

RUD(b1)
RUD(b2)
RUD(b3)
RUD(b4)
SID(n9)
SID(n8)
SID(n7)
RUD(n6)
SID(n5)
SID(n4)
RUD(n3)
RUD(n2)
!\textcolor{red}{$\UKD$(n1)}!

$\,$
\end{lstlisting}
          }
        \end{minipage}
        &
        \begin{minipage}{0.8\textwidth}
\begin{lstlisting}
bool mChi(bool i1, bool i2, bool i3, 
          bool r1, bool r2, bool r3)
{
    bool b1 = i1 $\oplus$ r1;
    bool b2 = i2 $\oplus$ r2;
    bool b3 = i3 $\oplus$ r3;
    bool b4 = b2 $\oplus$ b3;
    bool n9 = b3 $\land$ r2;
    bool n8 = r3 $\land$ r2;
    bool n7 = b2 $\lor$ r3;
    bool n6 = r1 $\oplus$ n9;
    bool n5 = n7 $\oplus$ n8;
    bool n4 = b2 $\lor$ b3;
    bool n3 = n5 $\oplus$ n6;
    bool n2 = n4 $\oplus$ b1;
    bool n1 = n2 $\oplus$ n3;
    return n1;
}
\end{lstlisting}
        \end{minipage}
        \\
        \cline{1-1}
      \end{tabular}
    }
    \label{sfig:motivation:example:program}
  }
  %
  %
  \subfloat[]
  {
  \adjustbox{raise=-6pc}{
    \scalebox{0.75}{
      \input{images/motivation/tree4.tex}
    }
    \label{sfig:motivation:example:tree4}
    }
  }
  \caption{The program on the left is a perfectly masked $\chi$
    function from MAC-Keccak.
    %
    The decision tree on the right represents the static analyzer that the user would like to synthesize. Here, $x$ is a program variable, whose
    type is being computed;  $L$ and $R$ are its
    left and right operands, respectively, and $f(x)$ is a synthesized
    feature shown in Fig.~\ref{sfig:motivation:setting:rules:learned}
    (represented by recursive Datalog program).
%
}
  \label{fig:motivation:example}
  \vspace{-1em}
\end{figure*}

Consider the program in Fig.~\ref{sfig:motivation:example:program},
which computes the $\chi$ function from Keccak, a family of
cryptographic primitives for the SHA-3
standard~\cite{nist2012keccak,bertoni2012keccak}. It ultimately
computes the function 
\lstinline|n1 = i1 $\oplus$ ($\lnot$ i2 $\land$ i3)|,  
where $\oplus$ means XOR.
Unfortunately, a straightforward implementation could potentially leak
knowledge of the secret inputs \lstinline|i1|, \lstinline|i2| and
\lstinline|i3| if the attacker were able to guess the intermediate
results \lstinline|$\lnot$ i2| and \lstinline|$\lnot$ i2 $\land$ i3|
via the \emph{power}
side-channels~\cite{eldib2014synthesis,barthe2018maskverif}.  The
masking countermeasures in the implementation therefore use three
additional random bits \lstinline|r1|, \lstinline|r2| and
\lstinline|r3| to prevent exposure of the private inputs while still
computing the desired function.

\setlength{\textfloatsep}{10pt plus 1.0pt minus 2.0pt}
\begin{figure}[t]
 \centering
\footnotesize{
  \subfloat[Excerpt of rules learned by the \GPS{} tool.]
  {
    $
      \begin{array}{rrcl}
        R_1: & \RUD(x) & \gets & \XOROP(x) \land \RCOP(x, R) \land \RUD(R) \land \lnot f(x) \\
        R_2: & \textcolor{black}{f(x)} & \gets & \LCOP(x, L) \land \RCOP(x, R) \land g_1(L, r_L) \land \\
             &      &        & g_2(R, r_R) \land r_L = r_R \\
        R_3: & g_1(r, r) & \gets & \INRANDOP(r) \\
        R_4: & g_1(x, r) & \gets & \LCOP(x, y) \land g_1(y, r) \\
        R_5: & g_1(x, r) & \gets & \RCOP(x, y) \land g_1(y, r) \\
        R_6: & g_2(r, r) & \gets &\INRANDOP(r) \\
        R_7: & g_2(x, r) & \gets & \LCOP(x, y) \land g_2(y, r) \land \XOROP(x) \\
        R_8: & g_2(x, r) & \gets & \RCOP(x, y) \land g_2(y, r) \land \XOROP(x) 
      \end{array}
    $ 
    \label{sfig:motivation:setting:rules:learned}
  }
  \vspace{1ex}
    
  \subfloat[Corresponding expert written rules from SCInfer~\cite{zhang2018sc}.]
  {
    $
      \begin{array}{rrcl}
        M_1: & \RUD(x) & \gets & \XOROP(x) \land \DOMOP(x, \text{res}) \land \text{res} \neq \emptyset \\
        M_2: & \SUPPOP(x, x) & \gets & \INRANDOP(x) \lor \INKEYOP(x) \lor \INPUBOP(x) \\
        M_3: & \SUPPOP(x, \text{res}) & \gets & \LCOP(x, L) \land \RCOP(x, R) \land \SUPPOP(L, S_L) \land \\
             &  &  & \SUPPOP(R, S_R) \land \text{res} = S_L \union S_R \\
        M_4: & \DOMOP(x, x) & \gets & \INRANDOP(x) \\
        M_5: & \DOMOP(x, \emptyset) & \gets & \INKEYOP(x) \lor \INPUBOP(x) \\
        M_6: & \DOMOP(x, \text{res}) & \gets & \XOROP(x) \land \LCOP(x, L) \land \RCOP(x, R) \land \\
             &    &    & \DOMOP(L, S_{DL}) \land \DOMOP(R, S_{DR}) \land \\
             &    &    & \SUPPOP(L, S_L) \land \SUPPOP(R, S_R) \land \\
             &    &    & \text{res} = (S_{DL} \union S_{DR}) \setminus (S_L \union S_R)
      \end{array}
    $
    \label{sfig:motivation:setting:rules:scinfer}
  }
}
  \caption{Comparing the rules learned by \GPS{}
    (Fig.~\ref{sfig:motivation:setting:rules:learned}) to manually
    crafted rules from SCInfer
    (Fig.~\ref{sfig:motivation:setting:rules:scinfer}). 
%
%
    Observe that the learned rules are \emph{sound}, i.e., every variable which potentially leaks information is assigned the distribution type $\UKD$, while still managing to draw non-trivial conclusions such as $\RUD(\mathlst{b4})$.
    \textcolor{black}{The learned rules ($R_2$---$R_8$) in Fig.~\ref{sfig:motivation:setting:rules:learned} are used
     to define the new feature $f(x)$ in Fig.~\ref{sfig:motivation:example:tree4}}.
    }
  \label{fig:motivation:setting:rules}
\end{figure}

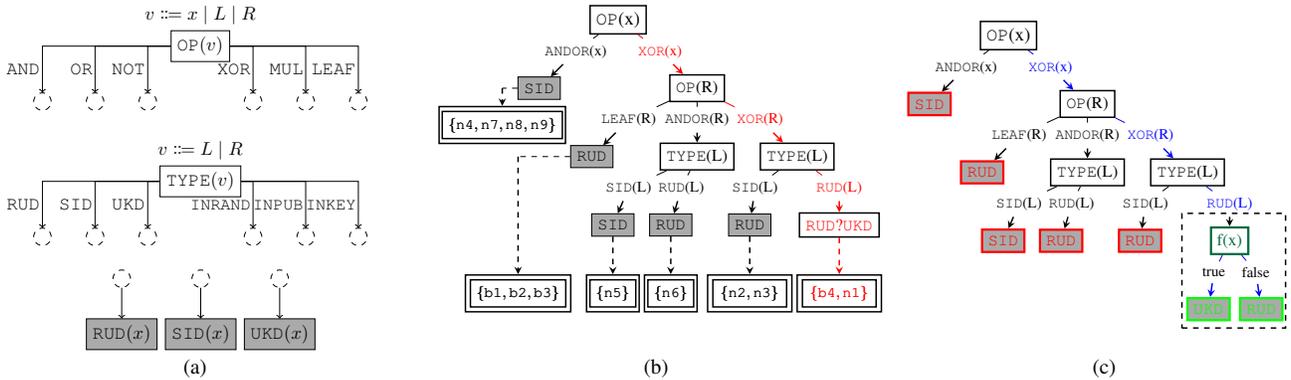
\begin{figure*}
\centering
\subfloat[]
{
  \begin{minipage}[c][1\width]{0.32\textwidth}
	   \centering
           \input{images/motivation/features.tex}
           \vspace{-1.5ex}
	\vspace{-1cm}
	\end{minipage}
  \label{sfig:motivation:synthesis:features:predefined}
}
\hfill
\hspace{-1.25em}
\subfloat[] 
{
  \begin{minipage}[c][1\width]{0.32\textwidth}
    \centering
    \scalebox{0.7}{\input{images/motivation/tree1.tex}}
    \vspace{-1.5ex}
	\vspace{-1cm}
  \end{minipage} 
  \label{sfig:motivation:synthesis:features:tree1}
}
\hfill
\hspace{-1.75em}
\subfloat[] 
{
  \begin{minipage}[c][1\width]{0.32\textwidth}
	   \centering
	   \scalebox{0.7}{\input{images/motivation/tree2.tex}}
	   \vspace{-1.5ex}
	\vspace{-1cm}
	\end{minipage}
  \label{sfig:motivation:synthesis:features:tree2}
}
\caption{
%
%
  The classifier of
  Fig.~\ref{sfig:motivation:synthesis:features:tree1} is learned
  only using the features in
  Fig.~\ref{sfig:motivation:synthesis:features:predefined}.  Because
  of the limited expressive power of these features, the learned
  analysis necessarily misclassifies either \lstinline|b4| or
  \lstinline|n1|. 
  Fig.~\ref{sfig:motivation:synthesis:features:tree2}
  denotes the candidate analyzer produced after one round of
  feature synthesis. The blue paths corresponds to the rule $\RUD(x) 
  \gets \XOROP(x) \land \XOROP(R) \land \RUD(L) \land \lnot f(x) 
  \land LC(x, L) \land RC(x, R)$. Unfortunately, even though 
  this analysis (Fig.~\ref{sfig:motivation:synthesis:features:tree2}) achieves 100\%
  training accuracy, the leaf nodes highlighted in red correspond to
  unsound predictions.  
  }
\label{fig:motivation:synthesis:features}
\vspace{-1em}
\end{figure*}

\subsection{Problem Setting}
\label{sub:motivation:setting}

Given such a masked program, users want to determine if they succeed
in eliminating side-channel vulnerabilities: in particular, if each
intermediate result is uniformly distributed ($\RUD$) or at least
independent of the sensitive inputs ($\SID$). The desired static
analysis thus associates each variable $x$ (e.g., $n_1$) with the
elements of a three-level abstract domain, $\RUD$, $\SID$ or $\UKD$,
indicating that $x$ is uniformly distributed ($\RUD$), secret
independent ($\SID$), or unknown ($\UKD$) and therefore potentially
vulnerable.

The decision tree in Fig.~\ref{sfig:motivation:example:tree4}
represents the desired static analyzer, which accurately classifies
most variables of the training corpus, and is also sound when applied
to new programs. Given variable $x$, the decision tree leverages the
features of $x$---such as the operator type of $x$ ($\OPOP$(x) :=
$\ANDOROP$(x)$|\XOROP$(x)) and the types of $x$'s operands
(e.g. $\TYPE$(L), $\TYPE$(R))---and maps $x$ to its corresponding
distribution type.  The white nodes of
Fig.~\ref{sfig:motivation:example:tree4} represent pre-defined
features, while the grey nodes represent output classes (associated
types).  Each path from the root to leaf node corresponds to one
analysis rule.  The set of pre-defined features used in this work is
shown in Fig.~\ref{sfig:motivation:synthesis:features:predefined}.

Designing side-channel analyses has been the focus of intense
research, see for example~\cite{wang2019mitigating, zhang2018sc,
barthe2018maskverif, wang2017security, barthe2016strong,
barthe2015verified,chen2017precise,tizpaz2019quantitative}.
Unfortunately, it requires expert knowledge in both computer security
and program analysis, and invariably involves delicate trade-offs
between accuracy and scalability. Our goal in this work is to assist
the analysis designer in automating the development.  This problem has
also been the subject of exciting research~\cite{david2015using,
bielik2017learning}; however, these approaches typically either
require computationally intensive deductive synthesis or cannot
guarantee soundness and thus produce errors in both directions,
including false alarms and missed bugs.

In contrast, \GPS{} combines inductive synthesis from user annotations
with logical entailment checking against a more comprehensive,
known-to-be-correct set of proof rules that form the knowledge base
($KB$). It takes as input training programs like the one in
Fig.~\ref{sfig:motivation:example:program}, where the labels
correspond to the types of program variables ($\RUD$/$\SID$/$\UKD$ for
intermediate results and $\INRANDOP$/$\INPUBOP$/$\INKEYOP$ for
inputs). The users are free to annotate as many or as few of these
types as they wish: this affects only the precision of the learned
analyzer and not its soundness.  Second, \GPS{} also takes as input
the knowledge base $KB$, consisting of proof rules that describe
axioms of propositional logic (Fig.~\ref{boolean}) and properties of
the distribution types (Fig.~\ref{fig:kbRule}). In return, \GPS{}
produces as output a set of Datalog rules which simultaneously
achieves high accuracy on the training data and provably sound with
respect to $KB$.

The proof rules for $KB$ 
were collected from published papers on masking
countermeasures~\cite{zhang2018sc, wang2019mitigating,
barthe2015verified}.  We emphasize that $KB$ is not necessarily an
executable static analyzer since repeated application of these proof
rules need not necessarily reach a fixpoint and terminate in finite
time; 
Furthermore, even in cases
where it does terminate, $KB$ may be computationally expensive and
infeasible for application to large programs.

As a concrete example, we compare excerpts of the rules learned by
\GPS{} in Fig.~\ref{sfig:motivation:setting:rules:learned} 
to the corresponding rules from SCInfer~\cite{zhang2018sc}---a
human-written analysis---in
Fig.~\ref{sfig:motivation:setting:rules:scinfer}.
%
$\LCOP(x,L)$ and $\RCOP(x, R)$ arises in both versions, indicating
that $L$ and $R$ are the left and right operands of $x$
respectively.  Specifically, in
Fig.~\ref{sfig:motivation:setting:rules:scinfer}, $\SUPPOP(x, y)$
indicates that $y$ is used in the computation of $x$ syntactically
while $ \DOMOP(x, y)$ denotes that random variable $y$ is used in the
computation of $x$ semantically.
Observe the computationally expensive set operations in the
human-written version to the simpler rules learned by \GPS{} without
loss of soundness or perceptible loss in accuracy. These points are
also borne out in our experiments in Table~\ref{tbl:FSE19}, where
SCInfer takes $>$45 minutes on some Keccak benchmarks, while our
learned analysis takes $<$5 seconds.

\GPS{} consists of two phases: First, it learns a set of
type-inference rules---alternatively represented either as Datalog
programs or as decision trees---that are consistent with the training
data. Second, it proves these rules against the knowledge base.
In the next two subsections, we will explain the  learning and
soundness proving processes respectively.

\subsection{Feature Synthesis and Rule Learning}
\label{sub:motivation:synthesis}


The learned analyzer associates each node $x$ of a program's abstract
syntax tree (AST) with an element of the distribution type $\{
\UKD(x), \SID(x), \RUD(x) \}$. 
We may therefore interpret the analyzer
as a decision tree that, by considering various features of an AST
node, maps it to a type. With a pre-defined set of features, such as
those shown in
Fig.~\ref{sfig:motivation:synthesis:features:predefined}, analyzers of
this form can be learned with classical decision tree learning (DTL)
algorithms.  Fig.~\ref{sfig:motivation:synthesis:features:tree1} shows
such an analyzer, learned from the labeled program of
Fig.~\ref{sfig:motivation:example:program}.

Unfortunately, the pre-defined features may not be strong enough
to distinguish between nodes with different training labels, e.g.,
\lstinline|b4| and \lstinline|n1| from the training program, which
have distinct labels $\RUD(\mathlst{b4})$ and $\UKD(\mathlst{n1})$,
but after being sifted into the node highlighted in red in
Fig.~\ref{sfig:motivation:synthesis:features:tree1}, cannot be
separated by any of the features from
Fig.~\ref{sfig:motivation:synthesis:features:predefined}. 
To ensure soundness, the learner would be forced to conservatively
assign the label $\UKD(x)$, which sacrifices the accuracy. 

\GPS{} thus includes a \emph{feature} synthesis engine, triggered whenever
the learner fails to distinguish between two differently labeled
variables. In tandem with recursive feature synthesis,
\GPS{} overcomes the limited expressiveness of DTL by enriching
syntax space to capture more desired patterns.
Observe that paths of a decision tree can
be represented as Datalog rules, e.g., the red path
in
Fig.~\ref{sfig:motivation:synthesis:features:tree1} is equivalent to
{
\[
  \UKD(x) \gets \XOROP(x) \land \XOROP(R) \land \RUD(L) \land \LCOP(x, L) \land \RCOP(x, R).
\]
}%
Viewing this in Datalog also allows us to conveniently describe
\emph{recursive} features, and reduce \emph{feature synthesis} to an
instance of syntax-guided synthesis (SyGuS).
Syntactically, the analysis rules corresponding to new features
are instances of a pre-defined set of \emph{meta-rules}, and the
target specification is to produce a Datalog program for a relation
$f(x)$ that has strictly positive information gain for the variables
under consideration (see Section~\ref{sec:learner} for details).

In our running example, the synthesizer produces the feature $f(x)$
shown in Fig.~\ref{sfig:motivation:setting:rules:learned}, which
intuitively indicates that some random input $r$ is used to compute
both operands of $x$.
With this new feature, the learner can distinguish between
\lstinline|b4| and \lstinline|n1|, and produce the rule shown in
Fig.~\ref{sfig:motivation:synthesis:features:tree2}, which correctly
classifies all variables of the training program. 
Observe that the rules defining $f(x)$ in
Fig.~\ref{sfig:motivation:setting:rules:learned} involve a newly
introduced predicate $g(x, r)$ and recursive structure that can
classify variables based on arbitrarily deep properties of the
abstract syntax tree. 
  
\begin{figure}
  \centering
  \scalebox{0.7}{
    \begin{tabular}{lcccccc}
      \toprule
      & \tikz{\node [root2] {$\OPOP(x)$}} & \tikz{\node [root2] {$\OPOP(L)$}} & \tikz{\node [root2] {$\OPOP(R)$}}
      & \tikz{\node [root2] {$\TYPE(L)$}} & \tikz{\node [root2] {$\TYPE(R)$}} & $f(x)$ \\
      \midrule
      $CE_1$ & $\ANDOROP$ & -1 & -1         & -1     & -1 & -1 \\
      $CE_2$ & $\XOROP$   & -1 & $\LEAFOP$  & -1     & -1 & -1 \\
      $CE_3$ & $\XOROP$   & -1 & $\XOROP$   & $\SID$ & -1 & -1 \\
      $CE_4$ & $\XOROP$   & -1 & $\ANDOROP$ & $\RUD$ & -1 & -1 \\
      $CE_5$ & $\XOROP$   & -1 & $\ANDOROP$ & $\SID$ & -1 & -1 \\
      \bottomrule
    \end{tabular}
  }
  %
  \caption{Abstract counter-examples produced during the soundness
    verification of the candidate analyzer shown in Fig.~\ref{sfig:motivation:synthesis:features:tree2}.
    }
    \label{sfig:motivation:proof:cex}
    \vspace{-1em}
\end{figure}

\setlength{\textfloatsep}{10pt plus 1.0pt minus 2.0pt}
\begin{figure}
  \centering
  \scalebox{0.8}{\input{images/motivation/tree3.tex}}
  \caption{Candidate analysis learned after one round of feedback from the soundness verifier. The leaves shown in green and red correspond to sound and unsound analysis rules respectively.}
  \label{fig:motivation:proof:tree3}
  \vspace{-1em}
\end{figure}
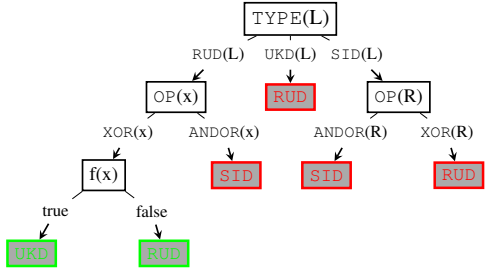

  %
  %

\subsection{Proving Soundness of Learned Analysis Rules}
\label{sub:motivation:proof}

While the learned analysis rules are \emph{correct by construction}
for the training examples, they may still be unsound when applied to
unseen programs. We observe this, for example, in the leaves
highlighted in red in
Fig.~\ref{sfig:motivation:synthesis:features:tree2}. 
Thus, \GPS{} tries to confirm their soundness against the
domain-specific knowledge base $KB$. 
In the context of our running example---confirming soundness means
proving that every variable $x$ that is assigned type $\RUD(x)$
(resp. $\SID(x)$) by the learned analysis rule $\alpha$ is also
certified $\RUD(x)$ (resp. $\SID(x)$) by $KB$.

We formalize the soundness proof as a Datalog query containment
problem,
%
%
and propose an algorithm based on bounded unrolling and
$k$-induction to check it.



When applied to the candidate analysis of
Fig.~\ref{sfig:motivation:synthesis:features:tree2}, the check results
in the five counter-examples $CE_1, \dots, CE_5$ with distribution
type $\UKD(CE_i)$ shown in Fig.~\ref{sfig:motivation:proof:cex}.  Each
counter-example indicates the unsoundness of one path from the root of
the decision tree to a classification node. These are \emph{abstract}
counter-examples in that they contain \emph{missing features} and
consequently do not define concrete ASTs. Thus, each of
these \emph{abstract} counter-examples is a set of feature valuations
$\pi = \{ f_1 \mapsto v_1, f_2 \mapsto v_2, \dots, f_k \mapsto v_k \}$
that the current candidate analysis misclassifies, and feeding them
back to the learner can prohibit subsequent candidate analyses from
classifying variables that satisfy $\pi$.

With these new constraints from abstract counter-examples, the learner
learns the new candidate analysis shown in
Fig.~\ref{fig:motivation:proof:tree3}. This new candidate analysis
still has four unsound candidate rules, which results in additional
abstract counter-examples when it is subjected to the soundness
check. We repeat this back-and-forth between the rule learner and the
soundness prover: after 11 iterations and after processing 27
counter-examples in all, \GPS{} learns the rules initially presented
in Fig.~\ref{sfig:motivation:example:tree4}, all of which have been
certified to be sound.

\subsection{Overall Architecture of the \GPS{} System}

We summarize the architecture of \GPS{} in Fig.~\ref{overview}.  The
learner repeatedly applies DTL and SyGuS to learn candidate analyses
that correctly classify training samples and are consistent with
newly-added abstract counter-examples. Next, the prover checks the
soundness of the learned analysis. Each subsequent counter-example is
fed back to the learner which restarts the rule learning process on
augmented dataset, until either all synthesized rules are sound or a
time bound is exhausted.


\section{Learning the Inference Rules} 
\label{sec:learner}

We formally describe the analysis rule learner in
Algorithm~\ref{alg:learner}. The input consists of a set of labeled
examples, $\mathcal{E}$, and a set of pre-defined features,
$\mathcal{F}$.  The output $\mathcal{T}$ is a set of type-inference
rules consistent with training examples.  Each training example
$(x, \TYPE(x)) \in \mathcal{E}$ consists of an AST node $x$ in a
program and its distribution type $\TYPE(x)$.

At the top level, the learner uses the standard \emph{decision tree
  learning (DTL)} algorithm~\cite{rokach2005top} as the baseline.
However, if it finds that the current set $\mathcal{F}$ of
classification features is insufficient, it invokes a
\emph{syntax-guided synthesis (SyGuS)} algorithm to synthesize a new
feature $f$ with strictly positive information gain to augment
$\mathcal{F}$.  This allows the learner to combine the efficiency of
techniques that learn decision trees with the expressiveness of syntax
guided synthesis; similar ideas have been fruitfully used in other
applications of program synthesis, see for
example~\cite{alur2017scaling}.

While the top-level classifier (e.g.,
Fig.~\ref{sfig:motivation:example:tree4},
\ref{sfig:motivation:synthesis:features:tree1},
\ref{sfig:motivation:synthesis:features:tree2} and
\ref{fig:motivation:proof:tree3}) has a bounded number of decision
points, the synthesized features (e.g.,
Fig.~\ref{sfig:motivation:setting:rules:learned}) may be
recursive. Furthermore, the newly synthesized features $f$ are
inducted as first-class citizens of $\mathcal{F}$, and can
subsequently be used at any level of the decision tree (see, for
example Fig.~\ref{sfig:motivation:example:tree4}
and~\ref{fig:motivation:proof:tree3}).
Next, we present the DTL and SyGuS subroutines respectively.

\begin{algorithm}[t]
\caption{$\DTL(\mathcal{E}, \mathcal{F})$ --- Decision Tree Learning.}
\label{alg:learner}
{\footnotesize
\input{images/learner/dtl.tex}
}
\end{algorithm}

\subsection{The Decision Tree Learning Algorithm}
\label{sub:learner:dtl}



Recall that our pre-defined features
(Fig.~\ref{sfig:motivation:synthesis:features:predefined}) include
properties of the AST node, such as $\OPOP(x)$, and properties
referring its left and right children, such as $\OPOP(L) \land
\LCOP(x, L)$. The choice requires some care: having very few features
will cause the learning algorithm to fail, while having too many
features will increase the risk of overfitting. Our synergistic
combination of DTL with SyGuS-based on-demand feature synthesis can be
seen as a compromise between these extremes.




$\DTL(\mathcal{E}, \mathcal{F})$ is an \emph{entropy-guided} greedy
learner~\cite{rokach2005top}, where the entropy and conditional
entropy of a set (defined below) are used to measure the diversity of
its labels:

{\footnotesize
\vspace{-1em}
$$
\begin{array}{l}
  H(\mathcal{E})  = -\sum_{t \in \TYPE} \Pr(\TYPE(x) = t) ~ \log(\Pr(\TYPE(x) = t)) \\\\
  H(\mathcal{E} \mid f)  = \sum_{i \in \Range(f)} H(\mathcal{E} \mid f(x) = i)
\end{array}
$$
}%

\noindent
%
Algorithm~\ref{alg:learner} thus divides the set of training examples
$\mathcal{E}$ using the feature $f = f^*$ that minimizes the conditional entropy $H(\mathcal{E} \mid f)$ (Lines~7--12), and recursively invokes the learning algorithm on each subset, $\DTL(\mathcal{E} \rvert_{f^*(x) = i},\mathcal{F} \setminus \{ f^* \})$.

Observe that $H(\mathcal{E}) = 0$ if $\Pr(\TYPE(x)=t)=100\%$, meaning
\emph{purity} or all examples $x \in \mathcal{E}$ share the same type
$\TYPE(x) = t$. The difference between $H(\mathcal{E})$ and
$H(\mathcal{E} \mid f)$ is also referred to as the \emph{information
  gain}. If the learner cannot find a feature with strictly positive
information gain (Line~4), it will invoke the feature synthesis
algorithm on Line~5.

\subsection{The On-Demand Feature Synthesis Algorithm}
\label{sub:learner:sygus}

\begin{algorithm}[t]
{\footnotesize
\input{images/learner/feature-syn.tex}
}
\caption{$\FeatureSyn$($\mathcal{E}$).}
\label{alg:learner:sygus}
\end{algorithm}

\setlength{\textfloatsep}{10pt plus 1.0pt minus 2.0pt}
\begin{figure}
{\footnotesize
  \[
    R_f = \left\{
            \begin{array}{rcl}
              f(x) & \gets & p_{in}(x), \\
              f(x) & \gets & q_{in}(x, y), \\
              f(x) & \gets & p_{in}(x, y) \land q_{in}(x, y), \\
              f(x) & \gets & q_{in}(x, y) \land f(y), \\
              f(x) & \gets & q_{in}(x, y) \land p_{in}(x) \land f(y)
            \end{array}
          \right\}
  \]

  \[
    R_g = \left\{
            \begin{array}{rcl}
              g(x, y) & \gets & q_{in}(x, y), \\
              g(x, y) & \gets & p_{in}(x) \land q_{in}(x, y), \\
              g(x, y) & \gets & q_{in}(x, z) \land g(z, y), \\
              g(x, y) & \gets & q_{in}(x, z) \land p_{in}(x) \land g(z, y)
            \end{array}
          \right\}
  \]

\[
  R_h = \left\{
          \begin{array}{rcl}
            h(x) & \gets & f(x) \land p_{in}(x) \land q_{in}(x, y), \\
            h(x) & \gets & g(x, y) \land p_{in}(x) \land q_{in}(x, y), \\
            h(x) & \gets & f(x) \land g(x, y) \land p_{in}(x) \land q_{in}(x, y)
          \end{array}
        \right\}
\]

\[
  \begin{array}{rcl}
    p_{in}(x)    & \Coloneqq & \ANDOP(x) \mid \OROP(x) \mid \NOTOP(x) \mid \XOROP(x) \mid \MULOP(x) \mid \LEAFOP(x) \\
                 & \mid      & \INRANDOP(x) \mid \INKEYOP(x) \mid \INPUBOP(x) \\
                 & \mid      & p_{in} \land p_{in} \mid p_{in} \lor p_{in} \mid \lnot p_{in} \\
    q_{in}(x, y) & \Coloneqq & \LCOP(x, y) \mid \RCOP(x, y) \mid x = y \\
                 & \mid      & q_{in}(x, y) \land q_{in}(x, y) \mid q_{in}(x, y) \lor q_{in}(x, y) \\
                 & \mid      & \lnot q_{in}(x, y)
  \end{array}
\]
}
\vspace{-1.5em}
\caption{Syntax of the DSL for synthesizing new features.}
\label{fig:learner:sygus:metarules}
\end{figure}

%
We represent newly synthesized features as Datalog programs. Datalog
is an increasingly popular medium to express static
analyses~\cite{Reps1995, bddbddb, jordan2016souffle, doop}, and its
recursive nature enables the newly learned features to represent
arbitrarily deep properties of AST nodes.
%
A Datalog rule is a constraint of the form
\begin{alignat}{1}
  h(\bm{x}) & \gets b_1(\bm{y}_1) \land b_2(\bm{y}_2) \land \dots \land b_n(\bm{y}_n), \label{eq:learner:sygus:rule}
\end{alignat}
where $h$, $b_1$\dots$b_n$ are relations with pre-specified arities and schemas, and where $\bm{x}$,
$\bm{y}_1$\dots$\bm{y}_n$ are vectors of typed variables. Each rule  
can be interpreted as a logical implication: if
$b_1$\dots$b_n$ are true, then so is $h$.
The semantics of a Datalog program is defined as the \emph{least fixed-point} of rule application~\cite{Alice}: the solver starts with
empty output relations, and repeatedly derives new output tuples until no new tuples can be derived.
%


%

\begin{figure*}
{\footnotesize
\begin{mathpar}
\inferrule*[right=($B_1$)] {b \vee \neg b \equiv true}{} \and
\inferrule*[right=($B_2$)] {b \wedge \neg b \equiv false}{} \and
\inferrule*[right=($B_3$)] {\neg\neg b \equiv b}{} \and
\inferrule*[right=($B_{4}$)] {\neg a \vee \neg b \equiv \neg(a \wedge b)}{}
\end{mathpar}
\begin{mathpar}
\inferrule*[right=($B_{5}$)] {\neg a \wedge \neg b \equiv \neg (a \vee b )}{} \and
\inferrule*[right=($B_6$)]{b \vee false \equiv b}{} \and
\inferrule*[right=($B_7$)]{b \vee true \equiv true}{} \and
\inferrule*[right=($B_8$)]{b \wedge true \equiv b}{} 
\end{mathpar}
\begin{mathpar}
\inferrule*[right=($B_9$)]{b \wedge b \equiv b}{}   \and
\inferrule*[right=($B_a$)]{b \wedge false \equiv false}{} \and
\inferrule*[right=($B_b$)]{b \vee b \equiv b}{}  \and
\inferrule*[right=($B_c$)]{a \wedge (a \vee b) \equiv a}{} 
\end{mathpar}
\begin{mathpar}
\inferrule*[right=($B_d$)]{a \vee(a \wedge b) \equiv a}{} \and
\inferrule*[right=($B_e$)]{a \oplus b \equiv (a \wedge \neg b) \vee ( \neg a \wedge b )}{} \and
\inferrule*[right=($B_f$)]{(a \vee b) \vee c \equiv a \vee c \vee b}{} 
\end{mathpar}
\begin{mathpar}
\inferrule*[right=($B_{10}$)]{(a \wedge b) \wedge c \equiv a \wedge c \wedge b}{} \and 
\inferrule*[right=($B_{11}$)]{a \vee (b \vee c) \equiv a \vee b \vee c}{} \and
\inferrule*[right=($B_{12}$)]{a \wedge (b \wedge c) \equiv a \wedge b \wedge c}{}
\end{mathpar}
}
\vspace{-1.5em}
\caption{Proof rules for propositional logic, to simplify the logic formula and deduce Boolean constants ($true$ and $false$).}
\vspace{-1em}
\label{boolean}
\end{figure*}
Program synthesis commonly restricts the space of 
target concepts and biases the search to speed up
computation and improve generalization. One form of 
bias has been to constrain the syntax: 
this has been formalized as the 
SyGuS problem~\cite{alur2013syntax} and as
meta-rules in inductive logic programming~\cite{Muggleton:MIL,
si2018syntax}. A meta-rule is construct of this form
%
\begin{alignat}{1}
  X_h(\bm{x}) & \gets X_1(\bm{y}_1) \land X_2(\bm{y}_2) \land \dots \land X_n(\bm{y}_n) \label{eq:learner:sygus:metarule}
\end{alignat}
Here, $X_h$, $X_1$, $X_2$, \dots, $X_n$ are \emph{relation variables} whose instantiation yields a concrete rule. 
Fig.~\ref{fig:learner:sygus:metarules} shows the meta-rules
used in our work. 
For example, instantiating the meta-rule $f(x) \gets
q_{in}(x, y) \land p_{in}(x) \land f(y)$ with $q_{in}(x, y) = \RCOP(x,
y)$ and $p_{in}(x) = \XOROP(x)$ yields  $f(x) \gets
\RCOP(x, y) \land \XOROP(x) \land f(y)$.
There are three variations of the final target relation schema,
$f(x)$, $g(x, y)$ and $h(x)$, where $x$ and $y$ denote AST nodes.

We formalize the synthesis problem as that of choosing a relation $R
\in \{ f(x), g(x, y), h(x) \}$ and finding a subset $P_D$ of its
instantiated meta-rules from Fig.~\ref{fig:learner:sygus:metarules}
such that the resulting Datalog program $P_D$ has strictly positive information
gain on the provided training examples $\mathcal{E}$.



\ignore{
In the subsequent soundness verification pass, our goal will be to determine if the learned analysis rules, represented as a Datalog program, is valid for all possible (unseen) programs.  While this problem is undecidable for general Datalog programs~\cite{calvanese2005decidable}, it does admit decision
procedures for certain restricted classes of programs, notably the UC2RPQ subset (Union of Conjunctive 2-Way Regular
Path Queries; see~\cite{calvanese2005decidable, calvanese2003reasoning, barcelo2014does}). We therefore carefully choose
the meta-rules in Fig.~\ref{fig:learner:sygus:metarules} to only produce concrete programs from the UC2RPQ fragment.
}






Algorithm~\ref{alg:learner:sygus} shows the procedure, which
repeatedly instantiates the meta-rules from
Fig.~\ref{fig:learner:sygus:metarules} and computes their
information gain. It successfully terminates when it discovers a
feature that can improve classification. Otherwise, it
returns failure (upon timeout) and invokes
$\DTL(\mathcal{E}, \mathcal{F})$ to conservatively classify
the decision tree node as being of type $\UKD$.



\begin{example}
\label{ex:learner:sygus}
Given $\mathcal{E} = \{ (\mathlst{b4}, \RUD), (\mathlst{n1}, \UKD) \}$
shown in Fig.~\ref{sfig:motivation:example:program}, the synthesizer
may alternatively learn the rules in
Equations~\ref{eq:learner:sygus:zeroGain},
\ref{eq:learner:sygus:valid} and~\ref{eq:learner:sygus:invalid}.

\vspace{-1em}
{\footnotesize
\begin{alignat}{2}
%
  f(x) & \gets \INRANDOP(x), \label{eq:learner:sygus:zeroGain} \\
  f(y) & \gets \LCOP(y, x) \land f(x), \nonumber \\
  f(y) & \gets \RCOP(y, x) \land f(x), \nonumber \\ 
  \RUD(x) & \gets \XOROP(x) \land \LCOP(x, L) \land \RCOP(x, R) \land \RUD(L) \land f(R). \nonumber
  \\[0.75em]
  g(x, x) & \gets \INRANDOP(x), \label{eq:learner:sygus:valid} \\
  g(y, z) & \gets \LCOP(y, x) \land g(x, z), \nonumber \\
  g(y, z) & \gets \RCOP(y, x) \land g(x, z), \nonumber \\
  h(x) & \gets \LCOP(x, L) \land \RCOP(x, R) \land g(L, x_L) \land g(R, x_R) \land x_L = x_R, \nonumber \\
  \RUD(x) & \gets \XOROP(x) \land \RUD(L) \land \RUD(R) \land \LCOP(x, L) \land \RCOP(x, R) \land \lnot h(x). \nonumber
  \\[0.75em]
  g(x, x) & \gets \INKEYOP(x), \label{eq:learner:sygus:invalid} \\
  g(y, z) & \gets \LCOP(y, x) \land g(x, z), \nonumber \\
  g(y, z) & \gets \RCOP(y, x) \land g(x, z). \nonumber \\
  h(x) & \gets \LCOP(x, L) \land \RCOP(x, R) \land g(L, x_L) \land g(R, x_R) \land x_L = x_R,  \nonumber \\
  \RUD(x) & \gets \XOROP(x) \land \RUD(L) \land \RUD(R) \land \lnot h(x). \nonumber
\end{alignat} \vspace{-1.5em}
}

Since the information gain of Rule~\ref{eq:learner:sygus:zeroGain}
applying to \{$\mathlst{b4}, \mathlst{n1}$\} is zero, it gets
discarded (Line~6 in Algorithm~\ref{alg:learner:sygus}).  In contrast,
the information gains of Rules~\ref{eq:learner:sygus:valid} and
\ref{eq:learner:sygus:invalid} are both positive.
Rule~\ref{eq:learner:sygus:valid} intuitively requires that both the
left and right operands of $x$ are of type $\RUD$, and that they do
not share any random inputs in computing $\lnot
h(x)$. Rule~\ref{eq:learner:sygus:invalid} requires that the same
secret key be used in the computation of both operands. 
While Rule~\ref{eq:learner:sygus:valid} is sound when applied to arbitrary
programs, Rule~\ref{eq:learner:sygus:invalid} is unsound. 
In
the next section, we will present an algorithm that can check the
soundness of these learned rules. 
\end{example}

\section{Proving the Inference Rules}
\label{sec:prover}


We wish to prove that a learned rule, denoted $\alpha$, never reaches
unsound conclusions when applied to any program, by showing that it
can be deduced from a \emph{known-to-be-correct} knowledge base
($KB$).  More specifically, we wish to confirm that every AST node $x$
marked as $\RUD$ (or $\SID$) by $\alpha$ can be certified to be $\RUD$
(or $\SID$) by $KB$.
When both $\alpha$ and $KB$ are expressed in Datalog, the problem reduces
to one of determining query containment, e.g., for every valuation of
the input relations, $\RUD_\alpha \subseteq \RUD_{KB}$ (or $\SID_\alpha \subseteq \SID_{KB}$).
We will now describe a semi-decision procedure to verify the soundness of the learned rules $\alpha$, which forms the second phase of the synthesis loop in  \GPS{}.

\subsection{Representation of the Learned Rule ($\alpha$)}

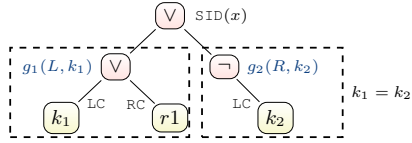
\begin{figure}
\centering
\scalebox{0.85}{\input{images/prover/AST_proof.tex}}
\vspace{-.5em}
\caption{Example $AST$ from which $\alpha$ is learned.}
\label{proofAST}
\vspace{-0.3em}
\end{figure}

Let $\alpha$ be a set of Datalog rules, each of which has a head
relation $\phi_{\alpha}$ and a body of the following form:
%
%
\begin{alignat}{1}
  \phi_{\alpha}(\bm{x}) & \gets \phi_1(\bm{x}_1) \land \phi_2(\bm{x}_2) \land \dots \land
                                   \phi_n(\bm{x}_n)
\end{alignat}
It means $\phi_{\alpha}$ holds only when all of $\phi_1,\ldots,\phi_n$ hold.
%
Here, $\phi_{\alpha}$ may be a distribution type, e.g., $\SID(x)$, or
a recursive feature $g(x,y)$, e.g.,  representing that $x$ depends on $y$.
%
%

\begin{figure*}
{\footnotesize
\begin{mathpar}
{\footnotesize
\inferrule*[right=($D_{1.1}$)]{\Gamma \vdash x : \INRANDOP}{supp(x, \{x\} )} \hspace{4ex}
\inferrule*[right=($D_{1.2}$)]{\Gamma \vdash x : \INKEYOP}{supp\text{(}x, \{x\} \text{)}} \hspace{4ex}
\inferrule*[right=($D_{1.3}$)]{\Gamma \vdash x : \INPUBOP}{supp\text{(}x, \{x\} \text{)}} \hspace{4ex}
\inferrule*[right=($D_{2.1}$)]{\Gamma \vdash x : \INRANDOP}{dom(x, \{x\})} \hspace{3ex}
\inferrule*[right=($D_{2.2}$)]{\Gamma \vdash x : \INKEYOP}{dom(x, \emptyset)}
}
\end{mathpar}
\begin{mathpar}
{\footnotesize
\inferrule*[right=($D_{1.4}$)]{\Gamma \vdash x,y : \mathsf{v},~ \Gamma \vdash S : \mathsf{Set}~\mathsf{v}, \RCOP (y, x_1) \wedge \LCOP ( y, x_2 ) \wedge supp(x_1, S_{1}) \wedge supp(x_2, S_{2})}{supp(y, S_1\cup S_2)} \hspace{10ex}
\inferrule*[right=($D_{2.3}$)]{\Gamma \vdash x : \INPUBOP}{dom\text{(}x, \emptyset \text{)}} 
}
\end{mathpar}
\begin{mathpar}
{\footnotesize
\inferrule*[right=($D_{2.4}$)]{\Gamma \vdash x,y : \mathsf{v}, \Gamma \vdash S : \mathsf{Set}~\mathsf{v},  \RCOP ( y, x_1 ) \wedge \LCOP ( y, x_2 ) \wedge \XOROP ( y ) \wedge dom(x_1, S_{1}) \wedge dom(x_2, S_{2})}{dom(x, (S_1 \cup S_2 )/(S_1 \cap S_2 ))}
}
\end{mathpar}
\begin{mathpar}
{\footnotesize
\inferrule*[right=($D_{3}$)]{\Gamma \vdash x : \mathsf{v}, \Gamma \vdash S : \mathsf{Set}~\mathsf{v}, dom(x, S_{x}) \wedge S_{x} \neq \emptyset}{\Gamma \vdash x : \RUD} \hspace{10ex}
\inferrule*[right=($D_{4}$)]{\Gamma \vdash x : \INKEYOP, \Gamma \vdash S :  \mathsf{Set}~\INKEYOP}{\Gamma \vdash x :: S : \mathsf{Set}~\INKEYOP} 
}
\end{mathpar}
\begin{mathpar}
{\footnotesize
\inferrule*[right=($D_{5}$)]{\inferrule{\Gamma \vdash x : \mathsf{v},~ \Gamma \vdash S_k : \mathsf{Set}~\INKEYOP, ~\\\\~ \Gamma \vdash S_d : \mathsf{Set}~\RUD,~ \Gamma \vdash S_s : \mathsf{Set}~\mathsf{v},~\\\\~ dom(x, S_{d}) \wedge S_{d} \text{=} \emptyset  \wedge supp(x, S_{s}) \wedge S_{s} \cap S_k \text{=} \emptyset}{} }{\Gamma \vdash x : \SID} \hspace{10ex}
\inferrule*[right=($D_{6}$)]{\inferrule {\Gamma \vdash x_1 : \SID,~\Gamma \vdash x_2 : \RUD,~\Gamma \vdash S_1, S_2 : \mathsf{Set}~\mathsf{v}~ \\\\ \LCOP ( y, x_1 ) \wedge \RCOP( y, x_2) \wedge \ANDOP ( y ) ~ \\\\ \wedge supp (x_1, S_1) \wedge supp (x_2, S_2) \wedge S_1 \cap S_2 \text{=} \emptyset}{} }{\Gamma \vdash y : \SID}
}
\end{mathpar}
\begin{mathpar}
{\footnotesize
\inferrule*[right=($D_{7}$)]{\inferrule {\Gamma \vdash x_1 : \SID,~\Gamma \vdash x_2 : \RUD,~\Gamma \vdash S_1, S_2 : \mathsf{Set}~\mathsf{v}~ \\\\  \LCOP ( y, x_1 ) \wedge \RCOP ( y, x_2) \wedge \OROP ( y ) \wedge \\\\ supp (x_1, S_1 ) \wedge supp (x_2, S_2) \wedge S_1 \cap S_2 \text{=} \emptyset}{} }{\Gamma \vdash y : \SID}  \hspace{10ex}
\inferrule*[right=($D_{8}$)]{\inferrule {\Gamma \vdash x_1 : \SID,~\Gamma \vdash x_2 : \SID,~\Gamma \vdash S_1, S_2 : \mathsf{Set}~\mathsf{v}~ \\\\ \LCOP \text{(} y, x_1 \text{)} \wedge \RCOP( y, x_2) \wedge \\\\  supp (x_1, S_1) \wedge supp (x_2, S_2) \wedge S_1 \cap S_2 \text{=} \emptyset}{} }{\Gamma \vdash y : \SID} 
}
\end{mathpar}
\begin{mathpar}
{\footnotesize
\inferrule*[right=($D_{9}$)]{\inferrule {\Gamma \vdash x_1 : \SID,~\Gamma \vdash x_2 : \SID,~\Gamma \vdash S_1: \mathsf{Set}~\RUD,~\Gamma \vdash S_2: \mathsf{Set}~\mathsf{v},\\\\ \ANDOP(y) \wedge  \LCOP ( y, x_1) \wedge \RCOP (y, x_2) \wedge   dom(x_1, S_1) \wedge supp(x_2, S_2)\wedge S_1 \cap S_2 \neq \emptyset}{} }{\Gamma \vdash y : \SID}
}
\end{mathpar}
\begin{mathpar}
{\footnotesize
\inferrule*[right=($D_{a}$)]{\inferrule {\Gamma \vdash x_1 : \SID,~\Gamma \vdash x_2 : \SID,~\Gamma \vdash S_1 : \mathsf{Set}~\RUD,~ \\\\~\Gamma \vdash S_2: \mathsf{Set}~\mathsf{v},~ \OROP(y) \wedge \LCOP \text{(} y, x_1 \text{)} \wedge \RCOP ( y, x_2) \wedge \\\\  dom (x_1, S_1) \wedge supp (x_2, S_2) \wedge S_1 / S_2 \neq \emptyset}{} }{\Gamma \vdash y : \SID}
\hspace{10ex}
\inferrule*[right=($D_{b}$)]{\inferrule {\Gamma \vdash x_1 : \RUD,~\Gamma \vdash x_2 : \RUD,~\Gamma \vdash S_1 : \mathsf{Set}~\RUD,~ \\\\ ~\Gamma \vdash S_2 : \mathsf{Set}~\mathsf{v},~\ANDOP \text{(} y \text{)} \wedge \LCOP \text{(} y, x_1 \text{)} \wedge \RCOP\text{(} y, x_2\text{)} \wedge \\\\  dom \text{(}x_1, S_1 \text{)} \wedge supp \text{(}x_2, S_2 \text{)} \wedge S_2 / S_1 \neq \emptyset}{} }{\Gamma \vdash y : \SID}
}
\end{mathpar}
\begin{mathpar}
{\footnotesize
\inferrule*[right=($D_{c}$)]{\inferrule {\Gamma \vdash x_1 : \RUD,~\Gamma \vdash x_2 : \RUD,~\Gamma \vdash S_1: \mathsf{Set}~\RUD,~\Gamma \vdash S_2: \mathsf{Set}~\mathsf{v},  \\\\  ~\OROP \text{(} y \text{)} \wedge ~\LCOP \text{(} y, x_1 \text{)} \wedge \RCOP\text{(} y, x_2\text{)} \wedge  dom \text{(}x_1, S_1 \text{)} \wedge supp \text{(}x_2, S_2 \text{)} \wedge S_2 / S_1 \neq \emptyset}{} }{\Gamma \vdash y : \SID}  
}
\end{mathpar}
\begin{mathpar}
{\footnotesize
\inferrule*[right=($D_{d}$)]{\Gamma \vdash x : \RUD}{\Gamma \vdash x : \texttt{NOUKD}} \hspace{1em} \hspace{5ex}
\inferrule*[right=($D_{e}$)]{\Gamma \vdash x : \SID}{\Gamma \vdash x : \texttt{NOUKD}} \hspace{1em} \hspace{5ex}
\inferrule*[right=($D_{f}$)]{\Gamma \vdash x : v,~\NOTOP \text{(} y \text{)} \wedge \LCOP \text{(}y, x\text{)}}{\Gamma \vdash y : v}
\hspace{5ex}
\inferrule*[right=($D_{10}$)]{\Gamma \vdash x : \mathsf{bool},~ x \text{=} \mathsf{true}}{\Gamma \vdash x : \SID}
}
\end{mathpar}
\begin{mathpar}
{\footnotesize
\inferrule*[right=($D_{11}$)]{\Gamma \vdash x : \mathsf{bool},~ x \text{=} \mathsf{false}}{\Gamma \vdash x : \SID} \hspace{10ex}
\inferrule*[right=($D_{12}$)]{{\inferrule{\Gamma \vdash x : \mathsf{v},~ \Gamma \vdash S_k : \mathsf{Set}~\INKEYOP, \Gamma \vdash S : \mathsf{Set}~\mathsf{v},~ supp\text{(}x, S_{s}\text{)} \wedge S_{s} \cap S_k \text{=} \emptyset}{} }
}{\Gamma \vdash x :  \texttt{NOUKD}} \vspace{1em}
}
\end{mathpar}
\begin{mathpar}
{\footnotesize
\inferrule*[right=($D_{13}$)]{\inferrule {\Gamma \vdash x_1 : \RUD,~\Gamma \vdash x_2 : \RUD,~\Gamma \vdash S_1, S_2 : \mathsf{Set}~\RUD,~ \\\\ 
\LCOP \text{(} y, x_1 \text{)} \wedge \RCOP\text{(} y, x_2\text{)} \wedge \MULOP(y) \wedge
\\\\
\text{(} y \text{)} \wedge dom \text{(}x_1, S_1 \text{)} \wedge dom \text{(}x_2, S_2 \text{)} \wedge S_2 / S_1 \neq \emptyset}{} }{\Gamma \vdash y : \SID}  
\hspace{10ex}
\inferrule*[right=($D_{14}$)]{\Gamma \vdash x_1 : \RUD,~\Gamma \vdash x_2 : \SID,~\Gamma \vdash S_1 : \mathsf{Set}~\RUD\\\\ \Gamma \vdash~S_2 : \mathsf{Set}~\mathsf{v},~ 
\LCOP \text{(} y, x_1 \text{)} \wedge \RCOP\text{(} y, x_2\text{)}~ \MULOP \text{(} y \text{)} \wedge
\\\\
dom \text{(}x_1, S_1 \text{)} \wedge supp \text{(}x_2, S_2 \text{)} \wedge S_1 / S_2 \neq \emptyset}{\Gamma \vdash y : \SID}
\hspace{1em}
}
\end{mathpar}
\begin{mathpar}
{\footnotesize
\inferrule*[right=($D_{15}$)]{\Gamma \vdash x_1 : \SID,~\Gamma \vdash x_2 : \RUD,~\Gamma \vdash S_1 : \mathsf{Set}~\RUD,~S_2 : \mathsf{Set}~\mathsf{v},~ \\\\ \LCOP \text{(} y, x_1 \text{)} \wedge \RCOP\text{(} y, x_2\text{)} \wedge \MULOP \text{(} y \text{)} \wedge dom \text{(}x_1, S_1 \text{)} \wedge supp \text{(}x_2, S_2 \text{)} \wedge S_2 / S_1 \neq \emptyset}{\Gamma \vdash y : \SID}
}
\end{mathpar}
}
\vspace{-2.15em}
\caption{Proof rules for distribution types, gathered from prior
works~\cite{zhang2018sc,eldib2014synthesis,barthe2015verified}. 
Here, $v$ denotes the type of variable $x$,and is of the following types: $\UKD$, $\SID$ and $\RUD$.
$\texttt{NOUKD}$ denotes the secure type (either $\RUD$ 
or $\SID$). All the predefined relations in $KB$ are the same as in $\alpha$.
}
\label{fig:kbRule}
\vspace{-1em}
\end{figure*}

\subsection{Representation of the Knowledge Base ($KB$)}
\label{sub:prover:kb}

Our $KB$ consists of two sets of proof rules, one for propositional
logic and the other for distribution types.

\vspace{1ex}
\noindent
\emph{Proof Rules for Propositional Logic.}
Fig.~\ref{boolean} shows the proof rules that represent axioms of
propositional logic~\cite{russell2002artificial}; they can be used to reduce any valid
(resp. invalid) Boolean formula to constant $true$ (resp. $false$).
Thus, they are useful in showing results such as $true \vee P$ and
$false \wedge Q$ are secret-independent ($\SID$), for arbitrary
logical sentences $P$ and $Q$.

Consider the example rule $\alpha$ below, where $g_1$ and $g_2$ are
synthesized features shown as dashed boxes in Fig.~\ref{proofAST}:
\begin{center}
\begin{tabular}{@{\hskip3pt}c@{\hskip3pt} c @{\hskip3pt}l@{\hskip3pt}}
{\footnotesize $\SID$($x$)} & {\footnotesize$\leftarrow$} & {\footnotesize $\OROP$($x$) $\wedge$ $\LCOP$($x$, $L$) 
$\wedge$ $\RCOP$($x$, $R$) $\wedge$ $\OROP$($L$)
$\wedge$ $\NOTOP$($R$) $\wedge$ } \\
 & & {\footnotesize $g_1$($L$, $k_1$)
$\wedge$ $g_2$($R$, $k_2$) $\wedge$ $\texttt{EQ}(k_1,k_2)$} \\
{\footnotesize $g_1$($L$, $k_1$)}& {\footnotesize $\leftarrow$} & 
{\footnotesize $\INKEYOP$($k_1$) $\wedge$ $\INRANDOP$($r_1$) $\wedge$ $\LCOP$($L$, $k_1$) $\wedge$ $\RCOP$($L$, $r_1$)} \\
{\footnotesize $g_2$($R$, $k_2$)} & {\footnotesize $\leftarrow$} & 
{\footnotesize $\INKEYOP$($k_2$) $\wedge$ $\LCOP$($R$, $k_2$) }
\end{tabular}
\end{center}
\noindent
Since $k_1=k_2$, we transform $\alpha$ into an equivalent logic formula: 
\begin{center}
{\footnotesize
$\SID(x) \leftarrow \texttt{EQ}(x, (k_1 \vee r_1) \vee (\neg k_1)) $ 
}
\end{center}
Rules $B1$, $B7$ and $Bf$ in Fig.~\ref{boolean} show that
$(k_1 \vee r_1) \vee (\neg k_1)$ is always $true$.  Thus, $x$ is
always $true$.  Since $x$ is a constant, we have $\SID$($x$), meaning $x$
is secret-independent.

Such $\SID$ rules, learned by our method automatically, and yet
overlooked by state-of-the-art, hand-crafted
analyzers~\cite{zhang2018sc,wang2019mitigating}, can significantly
improve the accuracy of side-channel analysis on many programs.

\ignore{
Note that repeated application of the rules shown in
Fig.~\ref{boolean} is not guaranteed to terminate.  The analysis
basically infers more equality relations by repeatedly applying the
equality analysis rules in Fig.~\ref{boolean}, with the help of
saturation engine~\cite{tate2009equality,kovacs2013first}.
Non-termination may happen when the equality saturation engine never
gets out of the equality analysis.  The axioms $B_6$: $b \vee false
\equiv b$ used in the direction from right to left can be applied
infinite times, generating successively larger expressions $(b, b \vee
false, (b \vee false) \vee false, ...)$.  As opposed to using
Fig.~\ref{boolean} as analyzer, here we leverage the rules of
Fig.~\ref{boolean} as the \emph{prover}, we are trying to prove the
following rule:
\begin{alignat}{1}
  \SID_{\alpha}(\bm{x}) & \gets \phi_1(\bm{x}_1) \land \phi_2(\bm{x}_2) \land \dots \land
                                   \phi_n(\bm{x}_n)
\end{alignat}
In the above rule, assuming that $\mathcal{S}$ represents the set of clauses in the body and more clauses
get added to $\mathcal{S}$ after applying equality analysis rules. 
As long as that there exists one constant clause ($true,
false$) transformed from the body of the original rule, the \emph{prover}
can terminate and validate the soundness. Otherwise, it would transfer
the soundness checking to the more comprehensive distribution rules
after hitting the timing bound.
}

%
\ignore{
%
While it is straightforward to add new rules to $KB$, in general, it
is difficult to know how useful each rule is in the analyzer.  
While adding too many inference rules to the analyzer is risky, e.g.,
it may drastically slow down the analyzer without bringing noticeable
increase in accuracy, adding them to $KB$ does not have such risk,
since $KB$ is used to prove the validity of $\alpha$ only during the
learning phase.
}

\vspace{1ex}
\noindent
\emph{Proof Rules for Distribution Types.}  
Fig.~\ref{fig:kbRule}
shows the proof rules that represent properties of the distribution
types.  They were collected from published
papers~\cite{zhang2018sc,eldib2014synthesis,barthe2015verified} that
focus on verifying masking countermeasures, which also provided the
soundness proofs of these rules.  For brevity, we omit the detailed
explanation. Instead, we use Rule~$D_{2.1}$ as an example to
illustrate the rationale behind these proof rules.

In Rule $D_{2.1}$, the $dom$($x$, $S$) relation means that variable $x$ is masked by
some input from the set $S$ of random inputs.
For example, in $y = x_1 \oplus x_2$, where $x_1 = k \oplus r_1 \oplus
r_2$ and $x_2 = b \oplus r_2$, we say that $x_2$ is masked by $r_2$,
and $x_1$ is masked by both $r_1$ and $r_2$.
However, since $r_2\oplus r_2 = false$, $y$ is masked only by $r_1$.
Thus, $dom$($y$, \{$r_1$\}) holds, but 
$dom(y,\{r_2\})$ does not hold.
In this sense, Rule $D_{2.4}$ defines a \emph{masking set}. For $y$, it is $S_y$ = (\{$r_1$, $r_2$\} $\cup$ \{$r_2$\}) $\setminus$ (\{$r_1$,
$r_2$\} $\cap$ \{$r_2$\}) = \{$r_1$\}, which contains $r_1$ only.
The masking set defined by $D_{2.4}$ is useful in that, as long as the
set is not empty, the corresponding variable is guaranteed to be
of the $\RUD$ type.

\ignore{
Even though these distribution rules are tailored to correctly
and soundly deriving the comprehensive examples cases as
compared with rules in Fig.~\ref{boolean}, it still
remains constrained in deriving the constant case, such as the AST example
shown in Fig.~\ref{proofAST} where the root variable $x$
is inferred as $\UKD$ by merely applying Fig.~\ref{fig:kbRule},
since it's inapplicable to any distribution rule.
Therefore, our $KB$ encompasses propositional rules as well
as distribution rules, capable of verifying more $\SID$ rules
and $\RUD$ rules while maintaining soundness.
}

\subsection{Proving the Soundness of $\alpha$ Using $KB$}
\label{sub:prover:ind}

To prove that for every AST node $x$ marked as $\RUD_\alpha(x)$
(resp. $\SID_\alpha(x)$) by $\alpha$, it is also marked as
$\RUD_{KB}(x)$ (resp. $\RUD_{KB}(x)$) by $KB$, we show that the
following relation $Ind(x)$ is empty for any valuation of the input
relations:
\begin{alignat}{1}
  Ind(x) & \gets \phi_\alpha(x) \land \lnot \phi_{KB}(x) ~,
\end{alignat}
where the relation $\phi$ may be instantiated to either $\RUD$ or
$\SID$.
In theory, this amounts to proving \emph{query containment}, which is
undecidable for Datalog in general~\cite{calvanese2005decidable,
calvanese1998decidability}, but there is a decidable Datalog
fragment~\cite{calvanese2005decidable, calvanese2003reasoning,
barcelo2014does}, and our meta-rules in
Fig.~\ref{fig:learner:sygus:metarules}
produce rules in this fragment.

\ignore{

In the remainder of this section, we present a semi-decision procedure
to prove the emptiness of $Ind(x)$.


\paragraph{Derivation trees, and unrolling a Datalog program}

First, we observe that every tuple $t = \phi(x)$ produced by a Datalog program is associated with one or more
\emph{derivation trees}. 
The heights of these derivation trees correspond to the
depth of rule inlining at which the program discovers $t$. In particular, for each inlining depth $k \in \N$, each rule
$\phi_h(\bm{x}_h) \gets \phi_1(\bm{x}_1) \land \phi_2(\bm{x}_2) \land \dots \land \phi_n(\bm{x}_n)$ can be transformed to
\begin{alignat}{1}
  \phi^{(k + 1)}(\bm{x}_h) & \gets \phi^{(k)}_1(\bm{x}_1) \land \phi^{(k)}_2(\bm{x}_2) \land \dots \land
                                   \phi^{(k)}_n(\bm{x}_n), \label{eq:prover:ind:unrolling}
\end{alignat}
where each $\phi^{(k)}$ contains exactly those tuples $\phi^{(k)}(\bm{x})$ which have a derivation
tree of depth $k$. Observe that $\phi = \bigcup_{k \in \N} \phi^{(k)}$.



\paragraph{Proving entailment at each unrolling depth.}

Our insight is to prove that at each unrolling depth $k$, $\phi^{(k)}_\alpha \subseteq \phi^{(k)}_{KB}$. In other words,
we define the relation $Ind^{(k)}$ as follows:
\begin{alignat}{1}
  Ind^{(k)}(x) \gets \phi^{(k)}_\alpha(x) \land \lnot \phi^{(k)}_{KB}(x), \label{eq:prover:ind:approx}
\end{alignat}
and extrapolate from the emptiness of $Ind^{(k)}$ for each $k$:
\begin{proposition}
If $Ind^{(k)}(x)$ is an empty relation for each inlining depth $k \in \N$, then $Ind(x)$ is also an empty relation.
\end{proposition}
\begin{proof}
Assume otherwise, so the $Ind$ relation contains some AST node $x$. By definition then, the candidate analysis $\alpha$
derives $\phi_\alpha(x)$, say at proof tree depth $l$, while $KB$ does not derive $\phi_{KB}(x)$. It follows that
$\phi^{(l)}_\alpha$ also contains $x$, and that $\phi^{(l)}_{KB}$ is an empty relation. We conclude that $x$ is an
element of $Ind^{(l)}$, which contradicts our hypothesis that $Ind^{(l)}$ was empty.
\end{proof}

\begin{note}
The converse of the above proposition need not hold. In particular, it may be the case that $Ind(x)$ is empty, even
though $Ind^{(k)}$ is inhabited. This is a curious consequence of connecting the inlining depths of $\phi_\alpha$ and
$\phi_{KB}$ in Equation~\ref{eq:prover:ind:approx}: there may be a tuple $\phi^{(k)}_\alpha(x)$ which is derived by $KB$
at some other inlining depth $k' \neq k$. In that case, since $x$ would then be absent from $\phi^{(k)}_{KB}(x)$, it
would inhabit $Ind^{(k)}(x)$, but not occur in $Ind(x)$.
\end{note}

}



First, we observe that every tuple $t = \phi(x)$ produced by a Datalog program is associated with one or more
\emph{derivation trees}. 
The heights of these derivation trees correspond to the
depth of rule inlining at which the program discovers $t$. In particular, for each inlining depth $k \in \N$, each rule
$\phi_h(\bm{x}_h) \gets \phi_1(\bm{x}_1) \land \phi_2(\bm{x}_2) \land \dots \land \phi_n(\bm{x}_n)$ is transformed
into the rule:
\begin{alignat}{1}
  \phi^{(k + 1)}(\bm{x}_h) & \gets \phi^{(k)}_1(\bm{x}_1) \land \phi^{(k)}_2(\bm{x}_2) \land \dots \land
                                   \phi^{(k)}_n(\bm{x}_n), \label{eq:prover:ind:unrolling}
\end{alignat}
Our insight is to prove that at each unrolling depth $k$, we have $\phi^{(k)}_\alpha \subseteq \phi^{(k)}_{KB}$. Thus,
we define the relation $Ind^{(k)}$ as follows:
\begin{alignat}{1}
  Ind^{(k)}(x) \gets \phi^{(k)}_\alpha(x) \land \lnot \phi^{(k)}_{KB}(x), \label{eq:prover:ind:approx}
\end{alignat}
and prove the emptiness of $Ind(x)$ by $k$-induction~\cite{sheeran2000checking,
  donaldson2011software, gadelha2017handling}.

\begin{proposition}
If $Ind^{(k)}(x)$ is an empty relation for each depth $k \in \N$, then $Ind(x)$ is an empty relation.
\end{proposition}

Observe that unrolling the rules of a Datalog program to any specific depth yields a formula which can be interpreted
within propositional logic. For example, unrolling $f(x)$ from Equation~\ref{eq:learner:sygus:zeroGain} at depths $1$
and $2$ gives us
{
\begin{alignat*}{1}
  f^{(1)}(x) & = \INRANDOP(x), \text{ and} \\
  f^{(2)}(y) & = (\LCOP(y, x) \land f^{1}(x)) \lor (\RCOP(y, x) \land f^{1}(x)).
\end{alignat*}%
}%
For any specific value of $k$, we can therefore use an SMT solver to verify the emptiness of 
$Ind^{(k)}$.

For the induction step, in particular, we ask the SMT solver to check
if $Ind^{(k)}$ can be non-empty while the $i$ preceding relations
$Ind^{(k-1)}$, $\dots$, $Ind^{(k - i)}$ are assumed to be empty.
Here, $\phi^{(k)}_\alpha$ is expressed recursively using
$\phi^{(k-1)}_\alpha$, $\ldots$, $\phi^{(k-i)}_\alpha$ and induction
succeeds if there exists such a value for $i\in\N$.

Let $V^{(k)}$ be free variables introduced by unrolling the rules at
depth $k$. We assert the non-emptiness of $Ind^{(k)}$ below:
\begin{alignat}{1}
  \Phi^{(k)} & = \bigvee_{x \in V^{(k)}} Ind^{(k)}(x).
\end{alignat}
Thus, we formalize the induction step of the proof by constructing the following formula:
\begin{alignat}{1}
  \Psi^{(k)} & = \lnot \Phi^{(k-i)} \land \dots \land \lnot \Phi^{(k-1)} \land \Phi^{(k)}
\end{alignat}
%
%
\begin{proposition}
\label{prop:prover:ind:k}
If for some $i \in \N$, the relations $Ind^{(1)}$, \dots, $Ind^{(i)}$ are all empty (the base case), and the formula $\Psi^{(k)}$ as defined above is unsatisfiable (the induction step), then $Ind^{(k)}$ is empty for all $k \in \N$.
\end{proposition}
Starting from $i = 1$, we use the SMT solver to check Proposition~\ref{prop:prover:ind:k} for increasingly larger $i$ until a timeout is reached. If the SMT solver is ever successful in proving the proposition, it follows that
the learned rule $\alpha$ is sound.

\subsection{Generating Abstract Counter-Examples}
\label{sub:prover:cex}

When the proof fails, however, we need to prevent the same rule from
being learned again to guarantee progress. 
%
%
%
%
Let $\pi = \{ f_1 = v_1, f_2 = v_2, \dots, f_k = v_k \}$ be the feature valuation in the failing
rule $R_\pi$. We then construct the counter-example,
\begin{center}
{\normalsize
$CE_\pi  = \{ f \mapsto v \mid (f, v) \in \pi \} \union \{ f \mapsto -1 \mid f \in \mathcal{F} \setminus \pi \}$
}
\end{center}
with label $\UKD(CE_\pi)$. Recall that $\mathcal{F}$ is the set of all features currently under consideration. 
Therefore, the feedback $CE_\pi$ provided to
$\DTL(\mathcal{E}, \mathcal{F})$ is an \emph{abstract} counter-example, with all missing features $f \in \mathcal{F} \setminus
\pi$ set to the unknown value $-1$.

Consider the subsequent iteration of the decision tree learner, $\DTL(\mathcal{E} \union \{ CE_\pi \}, \mathcal{F})$.
Observe that whenever it is in a decision context which is also a prefix $\pi_{pre}$ of the  counter-example
$CE_\pi$, the information gain of each feature $f \in \pi$ is strictly less than that encountered in the previous
invocation. Therefore, at some level of the decision tree, it will either choose a different feature, or invoke the feature synthesis algorithm to grow $\mathcal{F}$. By formalizing this argument, we say that:
\begin{proposition}
Given a counter-example $CE_\pi$ to a learned rule $R_\pi$, the subsequent invocation of the learner $\DTL(\mathcal{E} \union \{ CE_\pi \}, \mathcal{F})$ is guaranteed to no longer produce $R_\pi$.
\end{proposition}



Before ending this section, we stress that the \emph{proof rules} in $KB$ should not be confused with
\emph{analysis rules} used in the learned analyzer, since they are way more
computationally expensive.
Consider Rule~$D_{1.4}$, whose Datalog encoding size for $supp(x,S)$
would be $|V| \times 2^{|IN|}$.  For the benchmark named B19 in
Table~\ref{tbl:stats}, it owns 1250 input variables and thereby
causing the exponential explosion with $2^{1250}$.
The learned rule $\alpha$, in contrast, is much cheaper since it does
not rely on these expensive set (union and intersection) operations.

\section{Experiments}
\label{sec:experiment}

Our experiments were designed to answer the following research
questions (RQs):
\begin{itemize}
\item RQ1: How effective is our learned analyzer in terms of the analysis speed and accuracy?
\item RQ2: How effective is our \GPS{} method for learning  inference rules from training data?
\item RQ3: How effective is our \GPS{} method for proving the learned inference rules?
\end{itemize}

We implemented \GPS{} in LLVM 3.6.  \GPS{} relies on LLVM to parse the
C programs and construct the internal representation (IR).  Then, it
learns a static analyzer in two steps.  The first step, which is
SyGuS-guided decision tree learning, is implemented in 4,603 lines of
C++ code.  The second step, which proves the learned inference rules,
is implemented using the Z3~\cite{de2008z3} SMT solver.  Furthermore,
the learned analyzer (for detecting power side channels in
cryptographic software) is implemented in LLVM as an optimization
($opt$) pass.
We ran all experiments on a computer with 2.9 GHz Intel Core i5 CPU
and 8 GB RAM.  

\subsection{Benchmarks} \label{subsec:benchmark}

Our benchmarks are 568 programs with 2,691K lines of C code in total.
They implement well-known cryptographic algorithms such as AES and SHA-3.
Some of these programs are hardened by countermeasures, such as
reordered MAC-Keccak computation ~\cite{bertoni2012keccak}, masked
AES~\cite{blomer2004provably,barthe2015verified}, masked S-box
calculation~\cite{coron2013higher} and masked
multiplication~\cite{rivain2010provably}, to eliminate power side-channel
leaks.

We partition the benchmarks into two sets: $D_{train}$ for \GPS{}, and
$D_{test}$ for the learned analyzer.
The training set $D_{train}$ consists of 531 small programs gathered
from various public sources, including byte-masked
AES~\cite{yao2018fault}, random reduction of
S-box~\cite{zhang2017further}, common shares~\cite{coron2016faster},
and leak examples~\cite{eldib2014synthesis}.  Each benchmark is a
pair, consisting of a program AST and its distribution type, i.e, the
ground truth annotated by developers.
%
The testing set $D_{test}$ consists of 37 large programs, whose
statistics (the number of lines of code and inputs labeled public,
private, and random) are shown in Table~\ref{tbl:stats}.  Since these
programs are large, it is no longer practical to manually annotate the
ground truth; instead, we relies on the results of published tools: a
(manually-crafted) static analyzer~\cite{wang2019mitigating} for
B1-B20 and a formal verification tool~\cite{zhang2018sc} for B21-B37.

\begin{table}
\caption{Statistics of the  benchmark programs in $D_{test}$. 
}
\vspace{-0.5em}
\label{tbl:stats}
\centering
\scalebox{0.75}{
\begin{tabular}{|l|c|c|c|c||l|c|c|c|c|}
\hline
Name\ \ \   & LoC   & ${I}_{pub~}$ & ${I}_{priv}$ & ${I}_{rand}$ &
Name\ \ \   & LoC   & ${I}_{pub~}$ & ${I}_{priv}$ & ${I}_{rand}$  \\ \hline\hline
B1  &11    &0   & 2   & 2     & B2  &12     &0  & 2  & 2    \\ \hline
B3  &12    &0   & 1   & 2     & B4  &25     &1  & 1  & 3    \\ \hline
B5  &25    &1   & 1   & 3     & B6  &32     &1  & 1  & 3    \\ \hline
B7  &81    &1   & 1   & 7     & B8  &84     &1  & 1  & 7    \\ \hline
B9  &104   &1   & 1   & 7     & B10 &964    &1  & 16 &32    \\ \hline
B11 &1,130 &1   & 16  & 32    & B12 &1,256  &0  & 25 &75    \\ \hline
B13 &2,506 &0   & 25  & 125   & B14 &3,764  &0  & 25 &175   \\ \hline
B15 &8,810 &0   & 25  & 349   & B16 &13,810 &0  & 25 &575   \\ \hline
B17 &18,858&0   & 25  & 775   & B18 &23,912 &0  & 25 &975   \\ \hline
B19 &30,228&0   & 25  & 1,225 & B20 &34,359 &16 & 16 &1,232 \\ \hline\hline

B21 &79    &0   & 16  & 16    & B22  &67     &0  & 8  &16    \\ \hline
B23 &21    &0   & 2   & 2     & B24   &23     &0  & 2  &2     \\ \hline
B25  &27    &0   & 1   & 2     & B26  &32     &0  & 2  &2     \\ \hline
B27  &40    &0   & 2   & 3     & B28  &59     &0  & 3  &4     \\ \hline
B29  &60    &0   & 3   & 4     & B30 &66     &0  & 3  &4     \\ \hline
B31  &66  &0 & 3 & 4 & B32   &426k       &288   &288    &3205      \\ \hline
B33  &426k      &288    &288     &3205       & B34   &426k       &288   &288    &3205      \\ \hline
B35  &429k      &288    &288     &3205       & B36   &426k      &288   &288    & 3205     \\ \hline
B37  & 442k     &288    &288     & 3205      &    &   &   &    &      \\ \hline
\end{tabular}
}
\vspace{-1ex}
\end{table}

\subsection{Performance and Accuracy of the Learned Analyzer}
\label{subsec:analyzer}

To demonstrate the advantage of our learned analyzer (answer RQ1), we
compared our learned analyzer with the two existing
tools~\cite{zhang2018sc,wang2019mitigating} on the programs in
$D_{test}$.  Only our analyzer can handle all of the 37 programs.
Therefore, we compared the results of our analyzer with the tool
from \cite{wang2019mitigating} on B1-B20, and with the tool
from \cite{zhang2018sc} on B21-B37.  The results are shown in
Table~\ref{tbl:FSE19} and Table~\ref{tbl:CAV18}, respectively.

\begin{table}
\caption{Comparing the learned analyzer with the tool from \cite{wang2019mitigating}. 
}
\vspace{-0.5em}
\label{tbl:FSE19}
\centering
\scalebox{0.73}{
\begin{tabular}{|l|r|c|rrr|c|rrr|}
\hline
\multirow{2}{*}{Name}  & \multirow{2}{*}{\# AST}&  \multicolumn{4}{c|}{Manually Designed Analyzer~\cite{wang2019mitigating}} & \multicolumn{4}{c|}{Our Learned Analyzer}   \\\cline{3-10}
       &        &  Time (s) & UKD & SID & RUD    &  Time (s)   & UKD & SID & RUD
\\ \hline\hline
B1     &7      &0.061      &4 &    0 &    22       &\textbf{0.001}     &4 &0 &22
\\ \hline
B2     &6      &0.105      &\textbf{7} &    \textbf{0} &    20       &\textbf{0.001}     &\textbf{6} &\textbf{1} &20
\\ \hline
B3     &8      &0.099      &1 &    3 &    31       &\textbf{0.001}      &1 &3 &31
\\ \hline
B4     &11      &0.208  &\textbf{6}&12&31  &\textbf{0.001}      &\textbf{17}&12&20
\\ \hline
B5     &11      &0.216  &\textbf{1}&10&29  &\textbf{0.001}      &\textbf{11}&10&19
\\ \hline
B6     &14      &0.276  &\textbf{1}&15&48  &\textbf{0.001}      &\textbf{8}&15&41
\\ \hline
B7     &39      &0.213      &2 &    25 &   151      &\textbf{0.002}      &2 &    25 &   151
\\ \hline
B8     &39      &0.147      &4 &    42 &   249      &\textbf{0.002}      &4 &    42 &   249
\\ \hline
B9     &47      &0.266      &2 &   61 &   153      &\textbf{0.001}       &2 &   61 &   153
\\ \hline
B10    &522     &0.550      &31 &    12 &  2334      &\textbf{0.008}      &31 &    12 &  2334
\\ \hline
B11    &522     &0.447      &31&    0 &  2334      &\textbf{0.029}      &31&    0 &  2334
\\ \hline
B12    &426     &0.619      &52 &  300 &  2062      &\textbf{0.001}      &52 &  300 &  2062
\\ \hline
B13    &827     &1.102      &49 &  600 &  4030      &\textbf{0.006}      &49 &  600 &  4030
\\ \hline
B14    &1,228   &1.998      &49 &  900 &  5995      &\textbf{0.065}      &49 &  900 &  5995
\\ \hline
B15    &2,832   &16.999     &49&  2,100 &13861      &\textbf{0.107}      &49&  2,100 &13861
\\ \hline
B16    &4,436   &24.801     &49 &3,300 &21,723     &\textbf{2.663}       &49 &3,300 &21,723
\\ \hline
B17    &6,040   &59.120     &49 &4,500 &29,587      &\textbf{1.956}      &49 &4,500 &29,587
\\ \hline
B18    &7,644   &121.000 
                            &47 &5,700 &37,449      &\textbf{3.258}     &47 &5,700 &37,449
\\ \hline
B19    &9,649   &202.000 
                            &49 &7200 &47,280      &\textbf{5.381}       &49 &7200 &47,280
\\ \hline
B20    &13,826  &972.000 
                            &127 &26,330 &38,070      &\textbf{3.650}    &127 &26,330 &38,070
\\ \hline
\end{tabular}}
\end{table}

In both tables, Columns~1-2 show the benchmark name and the number of
AST nodes.  Columns~3-6 show the existing tool's analysis time and
result, including a breakdown in three types.  Similarly,
Columns~7-10 show our learned analyzer's time and result.
Note that in~\cite{wang2019mitigating}, the $\UKD$/$\SID$/$\RUD$ numbers 
were the number of variables of the LLVM IR, and thus larger
than the number of variables in the original programs. To be
consistent, we compared with their results in the same \mbox{manner in
Table~\ref{tbl:FSE19}.}

The results in Table~\ref{tbl:FSE19} and Table~\ref{tbl:CAV18} show
that our learned analyzer is much faster, especially on larger
programs such as B20 (3.6 seconds versus 16 minutes).
%
The reason why our analyzer is faster is because the manually-crafted analyses
~\cite{wang2019mitigating,zhang2018sc}
rely on evaluating set-relations (e.g. difference and intersection 
of sets of random variables), 
whereas our DSL syntax is designed without set operations to infer 
the same types, thus leading to faster analyses.
Although in general the set operation-based algorithm is
more accurate, it has excessive computational overhead.
Moreover, it does not always improve precision in practice.
Furthermore, the method in \cite{zhang2018sc} uses an SMT solver-based model counting technique to infer leak-free variables,
 which is significantly more \mbox{expensive than type inference}. 

As shown in Table~\ref{tbl:FSE19} and Table~\ref{tbl:CAV18}, 
by learning inference rules from data, we can achieve almost the same 
accuracy as manually-crafted analysis~\cite{wang2019mitigating,zhang2018sc} 
while avoiding the huge overhead.
Given the same definitions of distribution types
($\UKD$, $\SID$ and $\RUD$), both our learned rules and
manually-crafted analysis rules~\cite{wang2019mitigating,zhang2018sc} 
 can infer the non-leaky patterns,  thus recognizing the variable 
 types correctly under most benchmarks in
Table~\ref{tbl:FSE19} and Table~\ref{tbl:CAV18},
%
except for B4-B6 and B30, where set operations are required to prove
the leak-freedom of some variables.
 Recall that losing accuracy here
indicates that our learned rules infer the types more conservatively,
without losing soundness.
Nevertheless, our analyzer also increased accuracy in some other cases
(e.g., B2), due to its deeper constant propagation (which led to the
proof of more $\SID$ variables) while the existing
tool~\cite{wang2019mitigating} failed to do so, and conservatively
marked them as $\UKD$ variables.

\begin{table}
\caption{Comparing the learned analyzer with SCInfer~\cite{zhang2018sc}.}
\vspace{-0.5em}
\label{tbl:CAV18}
\centering
\scalebox{0.75}{
\begin{tabular}{|l|r|c|rrr|c|rrr|}\hline
\hline
\multirow{2}{*}{Name} & \multirow{2}{*}{\# AST}  & \multicolumn{4}{c|}{The SCInfer Verification Tool~\cite{zhang2018sc}} & \multicolumn{4}{c|}{Our Learned Analyzer}
\\\cline{3-10}
    &      & Time (s)  & UKD&SID&RUD          &  Time (s)   & UKD&SID&RUD
\\ \hline\hline
B21  &32   & 0.390     & 16&0&16        & \textbf{0.005}      & 16&0&16
\\ \hline
B22  &24   & 0.570     & 8&0&16         & \textbf{0.002}      & 8&0&16
\\ \hline
B23  &6    & 0.010     & 0&0&6          & \textbf{0.001}      & 0&0&6
\\ \hline
B24  &6    & 0.060     & 0&0&6          & \textbf{0.001}      & 0&0&6
\\ \hline
B25  &8    & 0.250     & 0&2&6          & \textbf{0.001}      & 0&2&6
\\ \hline
B26  &9    & 0.160     & 2&3&4          & \textbf{0.002}      & 2&3&4
\\ \hline
B27  &11   & 0.260     & 1&5&5          & \textbf{0.001}      & 1&5&5
\\ \hline
B28  &18   & 0.290     & 3&4&11         & \textbf{0.003}      & 3&4&11
\\ \hline
B29  &18   & 0.230     & 2&4&11         & \textbf{0.002}      & 2&4&12
\\ \hline
B30  &28   & 0.340     & 2&6&20         & \textbf{0.001}      & \textbf{8}&0&20
\\ \hline
B31  &28   & 0.500     & 2&7&19         & \textbf{0.001}      & 2&7&19
\\ \hline
B32  &197k & 3.800     & 0&6.4k&190.4k  & \textbf{3.180}          & 0&6.4k&190.4k
\\ \hline
B33  &197k & 2,828.000 
                       & 4.8k&6.4k&185.6k  & \textbf{3.260}         & 4.8k&6.4k&185.6k
\\ \hline
B34  &197k & 2,828.000 
                       & 3.2k&6.4k&187.2k  & \textbf{3.170}          & 3.2k&6.4k&187.2k
\\ \hline
B35  &198k & 2,828.000 
                       & 1.6k&8k&188.8k  & \textbf{3.140}          &\textbf{3.2k}&8k&187.2k
\\ \hline
B36  &197k & 2,828.000 
                       & 4.8k&6.4k&185.6k  & \textbf{3.150}           & 4.8k&6.4k&185.6k
\\ \hline
B37  &205k & 2,828.000 
                       & 17.6k&1.6k&185.6k  & \textbf{3.820}         & 17.6k&1.6k&185.6k
\\ \hline
\end{tabular}
}
\end{table}

%
\subsection{Effectiveness of Rule Induction and Soundness Verification}
\label{subsec:learnerProver}

To answer RQ2 and RQ3, we collected statistics while applying \GPS{}
to the 531 small programs in $D_{test}$, as shown in 
Table~\ref{tbl:result-learning}.  In total, \GPS{}
took 30 iterations to complete the entire learning process.
Column 1 shows the iteration number and Column 2 shows the time taken
by the \emph{learner} and the \emph{prover} together.  Columns~3-6
show the number of inference rules learned during each iteration,
together with their types ($\UKD$, $\SID$, and $\RUD$).
Similarly, Columns~7-10 show the number of verified inference rules and their types.

The next two columns show the following statistics: (1) the size of
the learned decision tree (\# Tree$_{learn}$) in terms of the number
of decision nodes; (2) the number of counter-examples (CEX) added by
the prover (\# AST$_{CEX}$), which are added to the 531 original
training programs before the next iteration starts.
The last column shows the number of features generated by SyGuS; these
features are also added to the original feature set and then used by the
learner during the next iteration.

\begin{table*}
\caption{Decision Tree Learning with Feature Synthesis (Different Iterations with \#AST = 531).}
\vspace{-0.5em}
\label{tbl:result-learning}
\centering
\scalebox{0.73}{
\begin{tabular}{|c||c|c|c|c|c|c|c|c|c|c|c|c|}
\hline
\multirow{2}{*}{Iteration}  
& \multirow{2}{*}{Time (s) } 
& \multicolumn{4}{c|}{\# Rules Learned}
& \multicolumn{4}{c|}{\# Rules Verified}
& \multirow{2}{*}{\# Tree$_{learn}$}
& \multirow{2}{*}{\# AST$_{CEX}$} 
& \multirow{2}{*}{\# Feature$_{syn}$}   
\\\cline{3-6} \cline{7-10}

    &   & \ \ \ Total\ \ \  & \textit{UKD}  &\textit{SID}  &\textit{RUD}  
        & \ \ \ Total\ \ \  & \textit{UKD}  &\textit{SID}  &\textit{RUD}  
        &  & &
\\\hline\hline
  1 & 1.316 & 9  & 2 & 2 & 5 & 5 & 2 & 1& 2 & 23 & 4 & 5
 \\\hline
  2 & 0.775 & 8  & 2 & 2 & 4 & 4 & 2 & 1& 1 & 17 & 9 & 7
  \\\hline
  3 & 1.115 & 8  & 2 & 2 & 4 & 5 & 2 & 2 & 1 & 24 & 13 & 9
 \\\hline
  4 & 0.511 & 8  & 2 & 2 & 4 & 5 & 2 & 1 & 2 & 18 & 18 & 10
 \\\hline
  5 & 0.513 & 8  & 2 & 2 & 4 & 7 & 2 & 2 & 3 & 27 & 21 & 11
 \\\hline
  6 & 0.537 & 8  & 2 & 2 & 4 & 6 & 2 & 2 & 2 & 24 & 25 & 12
 \\\hline
  7 & 0.510 & 8  & 2 & 2 & 4 & 6 & 2 & 2 & 2 & 26 & 29 & 13
 \\\hline
  8 & 0.512 & 8  & 2 & 2 & 4 & 6 & 2 & 2 & 2 & 28 & 33 & 14
 \\\hline
  9 & 0.511 & 8  & 2 & 2 & 4 & 6 & 2 & 2 & 2 & 30 & 37 & 15
 \\\hline
  10 & 0.524 & 8  & 2 & 2 & 4 & 5 & 2 & 2 & 1 & 32 & 41 & 16
  \\\hline
  11 & 0.546 & 8  & 2 & 2 & 4 & 4 & 2 & 2 & 0 & 34 & 45 & 17
  \\\hline
  12 & 0.556 & 8  & 2 & 2 & 4 & 4 & 2 & 2 & 0 & 36 & 49 & 18
  \\\hline
  13 & 0.550 & 8  & 2 & 2 & 4 & 5 & 2 & 2 & 1 & 38 & 53 & 19
  \\\hline
  14 & 0.540 & 8  & 2 & 2 & 4 & 6 & 2 & 2 & 2 & 40 & 57 & 20
  \\\hline
  15 & 0.542 & 8  & 2 & 2 & 4 & 4 & 2 & 2 & 0 & 42 & 61 & 21
  \\\hline
  16 & 0.552 & 8  & 2 & 2 & 4 & 6 & 2 & 2 & 2 & 44 & 65 & 22
  \\\hline
  17 & 0.577 & 8  & 2 & 2 & 4 & 5 & 2 & 2 & 1 & 46 & 69 & 23
  \\\hline
  18 & 0.598 & 8  & 2 & 2 & 4 & 6  & 2 & 2 & 2 & 48 & 73 & 24
  \\\hline
  19 & 0.571 & 8  & 2 & 2 & 4 & 6 & 2 & 2 & 2 & 50 & 77 & 25
  \\\hline
  20 & 0.673 & 8  & 2 & 2 & 4 & 5 & 1 & 2 & 2 & 52 & 82 & 26
  \\\hline
  21 & 0.526 & 8  & 2 & 2 & 4 & 3 & 1 & 2 & 0 & 54 & 87 & 27
  \\\hline
  22 & 0.525 & 8  & 3 & 2 & 3 & 6  & 3 & 2 & 1 & 35 & 91 & 27
  \\\hline
  23 & 0.697 & 9  & 3 & 2 & 4 & 7 & 2 & 2 & 3 & 37 & 93 & 27
  \\\hline
  24 & 0.700 & 9  & 3 & 2 & 4 & 8 & 2 & 2 & 4 & 38 & 95 & 28
  \\\hline
  25 & 0.691 & 7  & 2 & 2 & 3 & 6 & 1 & 2 & 3 & 36 & 97 & 29
  \\\hline
  26 & 0.707 & 7  & 2 & 2 & 3 & 6 & 1 & 2 & 3 & 37 & 99 & 30
  \\\hline
  27 & 0.716 & 7  & 2 & 2 & 3 & 6 & 1 & 2 & 3 & 38 & 101 & 31
  \\\hline
  28 & 0.540 & 7  & 2 & 2 & 3 & 6 & 1 & 2 & 3 &  39 & 102 & 32
  \\\hline
  29 & 0.534 & 7  & 2 & 2 & 3 & 6 & 1 & 2 & 3 & 39 & 103 & 32
  \\\hline
  30 & 0.528 & 7  & 2 & 2 & 3 & 7 & 2 & 2 & 3 & 39 & 104 & 32
  \\\hline\hline
TOTAL & 18.693 & 237 & 63 & 60 & 114 & 167 & 54 & 57 & 56 & 1071 & 1833 & 622
  \\\hline
\end{tabular}
}
\vspace{-2ex}
\end{table*}

Results in Table~\ref{tbl:result-learning} demonstrate the
efficiency of both the learner and the prover.
%
%
Within the \emph{learner}, the number of rules produced in each
iteration remains modest (8 on average), indicating it has
successfully avoided overfitting.
%
%
%
This is because the SyGuS solver is biased toward producing small
features which, by \emph{Occam's razor}, are likely to generalize
well.
%
Furthermore, any learned analysis rules have to pass the soundness
check, and this provides additional assurance against overfitting to
the training data.
%
%
%
The \emph{prover} either quickly verifies a rule, or quickly drops it after
adding a counter-example to prevent it from being learned again.
%
In early iterations, about half of all learned rules can be proved,
but as more counter-examples are added, the quality of the learned
rules improves, and thus the percentage of proved rules also increases.

\subsection{Threats to Validity}

Our experimental evaluation focused on cryptographic software, which
is structurally simple and, unlike general-purpose software, does not
exercise complicated language constructs.
It is an interesting direction of
future work to extend our techniques to these more general classes of
software code.

A notable limitation in our work is the assumption of the knowledge
base (KB). While KB is readily available for our application
(side-channel analysis), for other applications, it might be
non-trivial to construct. Furthermore, an incorrect KB might
compromise the soundness of the learned rules, although in this work,
we have carefully mitigated this threat by curating the proof rules
from previous
papers~\cite{zhang2018sc,eldib2014synthesis,barthe2015verified} that
have themselves formally verified the validity of these proof rules.

\section{Related Work}
\label{sec:related}

\noindent
\textit{Generating Analyzers from Examples.}
While there are prior works on learning static
analyzers~\cite{bielik2017learning,zaheer2016learning}, they do not
guarantee soundness.  For example, the analyzer learned by Bielik et
al.~\cite{bielik2017learning} is sound with respect to programs in the
training set, not all programs written in the same programming
language (JavaScript).  They also need to manually modify the training
programs to generate counter-examples, while our method generates
counter-examples automatically.

\vspace{0.25ex}
\noindent
\textit{Formal Specifications.}
There are also works on synthesizing static analyzers from formal
specifications, e.g., proof rules or
second-order logic
formulas~\cite{david2015using,grebenshchikov2012synthesizing,
  chen2019vfql} as opposed to training data.
However, they restrict the logic used to write the specification, and
as a result, may not be expressive enough to synthesize practical
analyzers.  Users are also expected to write correct
specifications, which is a non-trivial task.  In addition, they cannot
exploit the information provided by data.

\vspace{0.25ex}
\noindent
\textit{Learning-based Techniques.}
There are several prior techniques using machine learning to
conduct static program
analyses~\cite{katz2016estimating,raychev2015predicting,
  tripp2014aletheia,gvero2015synthesizing}.
Such techniques focus on finding a suitable program-to-feature
embedding. However, they require the user to perform feature
engineering, which is known to be laborious.
Some of these
techniques~\cite{mangal2015user,raychev2015predicting,heo2018adaptive,heo2016learning}
do not take advantage of new features that may be learned from
data; instead, they build classifiers based solely on existing
features.
In contrast, our method not only learn new analysis rules from data,
but also use SyGuS to synthesize new features automatically.

\vspace{0.25ex}
\noindent
\textit{Optimizing an Analyzer.}
It is possible to optimize an existing static
analyzer~\cite{grigore2016abstraction,
  heo2019resource,oh2015learning,heo2017machine,singh2018fast,katz2016estimating},
which can be achieved by adjusting the level of
abstraction~\cite{grigore2016abstraction, heo2019resource,wang2007using}, learn
heuristics and parameters~\cite{oh2015learning}, make
soundness-accuracy trade-offs~\cite{heo2017machine}, or select sound
transformers~\cite{singh2018fast}.
However, such techniques fundamentally differ from our method because
they assume the analyzer is already given, and focus on
optimizing its performance, whereas we focus on synthesizing a new
analyzer.

\vspace{0.25ex}
\noindent
\textit{Syntax-Guided Synthesis.}
Since we automatically generate new features, our method is related to
the large and growing body of work on SyGuS.
While SyGuS has been used in various
applications~\cite{rolim2017learning,mechtaev2016angelix,feng2017component,
  polikarpova2016program,abate2017automated,
  bavishi2019phoenix,xiong2017precise,le2017s3,hua2018towards,
  knoth2019resource,li2015resource}, none of them aims to synthesize a
provably sound static analyzer from data.
While some of these existing techniques can synthesize Datalog
rules~\cite{si2018syntax,zhang2014abstraction,si2019synthesizing}, the
focus has been on efficiency, e.g., pruning the search space based on
syntactic structures, instead of guaranteeing the soundness of the
analyzer.

\vspace{0.25ex}
\noindent
\textit{Power Side-Channel Analysis.}
In this work, we use power side-channel analysis as the application to
evaluate our method.  In this sense, it is related to the body of work
on side-channel leak
detection~\cite{chen2017precise,brennan2018symbolic,
bang2016string,zhang2018sc,guo2018adversarial,sung2018canal} as well as
mitigation~\cite{tizpaz2019quantitative,wang2019mitigating,
cauligi2019fact,wu2018eliminating,eldib2014synthesis,paulsen2019debreach}.
While static analysis engines used in these existing works are all
hand-crafted by domain experts, our method aims to synthesize the
static analysis from data automatically.

\section{Conclusions}
\label{sec:conclusion}

We have presented a data-driven method for learning a \emph{provably
sound} static analyzer to detect power side channels in cryptographic
software.  It relies on SyGuS to generate features and DTL to generate
analysis rules based on the synthesized features.  It verifies the
soundness of these learned analysis rules by solving a query
containment checking problem using an SMT solver.  We have evaluated
our method on C programs that implement well-known cryptographic
protocols and algorithms.  Our experimental results show that the
learning algorithm is efficient and the learned analyzer can achieve
the same empirical accuracy as state-of-the-art analysis tools while
being several orders-of-magnitudes faster.

\section*{Acknowledgments}

This research was supported in part by the U.S.\ National
Science Foundation (NSF) under grant CNS-1617203 and Office of
Naval Research (ONR) under grant \mbox{N00014-17-1-2896.} 
We thank the anonymous reviewers for their helpful feedback.
\clearpage
\bibliography{references}

\end{document}

%% file: images/introduction/jigsaw.tex
\begin{tikzpicture}[font=\scriptsize] 
 \tikzstyle{arrow1}=[thick,->,>=stealth,darkmain]
 \tikzstyle{arrow2}=[thick,->,>=stealth,black]

\tikzstyle{inBox}=[%
rectangle,draw, minimum width=1.5cm, minimum height=1.1cm, text width=1.3cm, inner sep = 0cm, outer sep = 0cm, align=center
]

\node at (0,0) [inBox, xshift=0.1mm, yshift=5mm] (b0) {Feature \\ Synthesis \\ (SyGuS)};
\node at (0,-1) [inBox, xshift=0.1mm, yshift=4mm] (b1) {Decision \\ Tree \\ Learning};
\node at (2,0) [inBox, xshift=5.1mm, yshift=5mm] (b2) {Query \\ Containment \\ Checking};
\node at (2,-1) [inBox, xshift=5.1mm, yshift=4mm] (b3) {Knowledge \\ Base \\ (KB)};


\tikzstyle{doc}=[%
minimum height=2em,
minimum width=4em,
draw, fill=white
]
   \node[doc] (doc) at (-2.5,0.7)  {};
   \node[doc, below left = 4pt and 4pt of doc.north east] (doc2) {};
   \node[doc, below left = 4pt and 4pt of doc2.north east] (doc3) {};

\tikzstyle{doc}=[%
minimum height=2em,
minimum width=4em,
draw, fill=white
]
   \node[doc] (doc) at (-2.5,-0.8)  {};
   \node[doc, below left = 4pt and 4pt of doc.north east] (doc2) {};
   \node[doc, below left = 4pt and 4pt of doc2.north east] (doc3) {};
   
\node at (-2.8,0.5) {Training};
\node at (-2.8,0.2) {Programs};

\node at (-2.8,-1) {Type};
\node at (-2.8,-1.2) {Annotations};

 \node[minimum width=17mm, minimum height=27mm, dashed, draw=darkmain, thick, align=center] (PR1) at (0.01, -0.2) {};   
 \node[text width=1.6cm, align=center] at (0.01, -1.35) {\textcolor{darkmain}{\textbf{Learner}}};
\node[minimum width=17mm, minimum height=27mm, dashed, draw=darkmain, thick, align=center] (PR2) at (2.5, -0.2) {};  
   \node[text width=1.6cm, align=center] at (2.5, -1.35) {\textcolor{darkmain}{\textbf{Prover}}};
   

  \tikzstyle{doc}=[%
minimum height=2em,
minimum width=3em,
draw, fill=white
]
   \node[doc] (doc) at (5.6,0.6)  {};
   \node[doc,align=center, inner sep=0cm, text width= 1.2cm, below left = 4pt and 4pt of doc.north east] (ANB) {Analyzer};

        \tikzstyle{doc}=[%
minimum height=2em,
minimum width=3em,
draw, fill=white
]
   \node[doc] (doc) at (5.6,-0.6)  {};
   \node[doc,align=center, inner sep=0cm, text width= 1.2cm, below left = 4pt and 4pt of doc.north east] (CEB) {Counter- \\ example};

   \node[rotate=22] at (4.05,0.4) {\color{darkmain} $R$ \textit{Verified}};
      \node[rotate=-25] at (4.05, -0.6) {\color{darkmain} $R$ \textit{Rejected}};
      
      \node  at (1.5, 0.2) {\textcolor{darkmain}{R}};
   \draw[arrow1,darkmain] (PR1.20) |- (PR2.170);
   \draw[arrow2,darkmain] (PR2.10) -- (ANB.west);
   \draw[arrow2,darkmain] (PR2.10) -- (CEB.west) ;
   
   \draw[arrow1,darkmain] (-1.8, 0.55) -- (-0.8, 0.55);
   
   \draw[arrow1,darkmain] (-1.8, -0.9) -- (-0.8, -0.9);

   \draw[arrow2,darkmain] (0.0, -2.0) -- (PR1.south);

   \draw [thick,darkmain] (CEB.south) -- (5.4, -2.0);
   \draw [thick,darkmain] (5.4, -2.0) -- (0.0, -2.0); 
\end{tikzpicture}

%% file: images/motivation/tree4.tex
\begin{tikzpicture}[font=\footnotesize]
 \tikzstyle{arrow1}=[thick,->,>=stealth,darkmain]
 \tikzstyle{arrow2}=[thick,->,>=stealth,black]
 \tikzstyle{root}=[rectangle,draw, align=center, thick]
 \tikzstyle{leaf}=[rectangle, draw, align=center, fill={rgb:black,1;white,2}]
 
 \node at (0,0) [root] (root0) {\normalsize{f(x)}};
 
 \node at (-.8, -.6) [] (rootConditionL) {\color{black} true};
 \node at (.8, -.6) [] (rootConditionR) {\color{black} false};
 
 \draw (root0) -- (rootConditionL);
 \draw (root0) -- (rootConditionR);
 
 \node at (-1.5, -1.3) [leaf] (branch11) {\small{$\UKD$}};
 \node at (1.5, -1.3) [root] (branch12) {\small{$\OPOP$(x)}};
 
 \draw [arrow2] (rootConditionL) -- (branch11);
 \draw [arrow2] (rootConditionR) -- (branch12);
 
 \node at (1.5-.8, -1.3-.6) [] (OPConditionL) {\color{black} $\XOROP(x)$};
 \node at (1.5+.8, -1.3-.6) [] (OPConditionR) {\color{black} $\ANDOROP(x)$};
 
 \draw (branch12) -- (OPConditionL);
 \draw (branch12) -- (OPConditionR);
 
 \node at  (1.5-.8-.7, -1.3-.6-.7) [root] (branch21) {\small{$\TYPE$(R)}};
 \node at (1.5+.8+.7, -1.3-.6-.7) [root] (branch22) {\small{$\TYPE$(L)}};
 
  \draw [arrow2] (OPConditionL) -- (branch21);
 \draw [arrow2] (OPConditionR) -- (branch22);

  \node at (1.5-.8-.7-.8, -1.3-.6-.7-.6) [] (TYPER_conditionL) {\color{black} $\SID$(R)};
  \node at (1.5-.8-.7+.8, -1.3-.6-.7-.6) [] (TYPER_conditionR) {\color{black} $\RUD$(R)};
  
  \draw (branch21) -- (TYPER_conditionL);
  \draw (branch21) -- (TYPER_conditionR);
  
  \node at (1.5-.8-.7-.8-.2, -1.3-.6-.7-.6-.7) [root] (branch31) {\small{$\TYPE$(L)}};
  \node at (1.5-.8-.7+.8+.2, -1.3-.6-.7-.6-.7) [leaf] (branch32) {\small{$\RUD$}};
  
  \draw [arrow2] (TYPER_conditionL) -- (branch31);
  \draw [arrow2] (TYPER_conditionR) -- (branch32);
  
  \node at (1.5+.8+.7-.8, -1.3-.6-.7-.6) (TYPEL_conditionL) {\color{black} $\neg\RUD$(L)};
  \node at (1.5+.8+.7+.8, -1.3-.6-.7-.6) (TYPEL_conditionR) {\color{black} $\RUD$(L)};
  
  \draw (branch22) -- (TYPEL_conditionL);
  \draw (branch22) -- (TYPEL_conditionR);
  
  \node at (1.5+.8+.7-.8-.2, -1.3-.6-.7-.6-.7) [leaf] (branch33) {\small{$\UKD$}};
  \node at (1.5+.8+.7+.8+.2, -1.3-.6-.7-.6-.7) [root] (branch34) {\small{$\TYPE$(R)}};
  
  \draw [arrow2] (TYPEL_conditionL) -- (branch33);
  \draw [arrow2] (TYPEL_conditionR) -- (branch34);
  
  \node at (1.5-.8-.7-.8-.2-.8, -1.3-.6-.7-.6-.7-.6) [] (TYPEL_conditionL2) {\color{black} $\RUD$(L)};
   \node at (1.5-.8-.7-.8-.2+.8, -1.3-.6-.7-.6-.7-.6) [] (TYPEL_conditionR2) {\color{black} $\neg\RUD$(L)};
   
  \draw (branch31) -- (TYPEL_conditionL2);
  \draw (branch31) -- (TYPEL_conditionR2);
  
   \node at (1.5-.8-.7-.8-.2-.8-.2, -1.3-.6-.7-.6-.7-.6-.7) [leaf] (branch41) {\small{$\RUD$}};
   \node at (1.5-.8-.7-.8-.2+.8+.2, -1.3-.6-.7-.6-.7-.6-.7) [leaf] (branch42) {\small{$\UKD$}};
   
   \draw [arrow2] (TYPEL_conditionL2) -- (branch41);
  \draw [arrow2] (TYPEL_conditionR2) -- (branch42);
  
  \node at (1.5+.8+.7+.8+.2-.8, -1.3-.6-.7-.6-.7-.6) [] (TYPER_conditionL2) {\color{black} $\neg\RUD$(R)};
  \node at (1.5+.8+.7+.8+.2+.8, -1.3-.6-.7-.6-.7-.6) [] (TYPER_conditionR2) {\color{black} $\RUD$(R)};
  
  \draw (branch34) -- (TYPER_conditionL2);
  \draw (branch34) -- (TYPER_conditionR2);
  
  \node at (1.5+.8+.7+.8+.2-.8-.2, -1.3-.6-.7-.6-.7-.6-.7) [leaf] (branch43) {\small{$\UKD$}};
  \node at (1.5+.8+.7+.8+.2+.8+.2, -1.3-.6-.7-.6-.7-.6-.7) [leaf] (branch44) {\small{$\SID$}};
  
     \draw [arrow2] (TYPER_conditionL2) -- (branch43);
  \draw [arrow2] (TYPER_conditionR2) -- (branch44);
  
  \node at (1.5-.8-.7-.8-.2-.8-.2-1.5, -1.3-.6-.7-.6-3) [root, minimum width=14mm, minimum height=5mm] (example1) {$\{ \mathlst{n1}, \mathlst{n5}\}$};
  
   \node at (1.5-.8-.7-.8-.2-.8-.2-1.5, -1.3-.6-.7-.6-3) [root, minimum width=16mm, minimum height=7mm] (example1Out) {$\{ \mathlst{n1}, \mathlst{n5}\}$};
  
  \node at (1.5-.8-.7-.8-.2-.8-.2, -1.3-.6-.7-.6-3) [root, minimum width=7mm, minimum height=5mm] (example2) {$\{ \mathlst{n6} \}$};
  \node at (1.5-.8-.7-.8-.2-.8-.2, -1.3-.6-.7-.6-3) [root, minimum width=10mm, minimum height=7mm] (example2Out) {$\{ \mathlst{n6} \}$};
  
   \node at (1.5-.8-.7+.8+.2, -1.3-.6-.7-.6-3) [root, minimum width=31mm, minimum height=5mm] (example3) {$\{ \mathlst{b1}, \mathlst{b2}, \mathlst{b3}, \mathlst{b4}, \mathlst{n2}, \mathlst{n3} \}$};
   
   \node at (1.5-.8-.7+.8+.2, -1.3-.6-.7-.6-3) [root, minimum width=34mm, minimum height=7mm] (example3Out) {};
   
   \node at (1.5+.8+.7+.8+.2+.8+.2, -1.3-.6-.7-.6-3) [root, minimum width=24mm, minimum height=5mm] (example4) {$\{ \mathlst{n4}, \mathlst{n7}, \mathlst{n8}, \mathlst{n9}\}$};
    \node at (1.5+.8+.7+.8+.2+.8+.2, -1.3-.6-.7-.6-3) [root, minimum width=26mm, minimum height=7mm] (example4Out) {};
  
  \draw [dashed] (branch11) -- (1.5-.8-.7-.8-.2-.8-.2-1.5, -1.3);
  \draw [arrow2, dashed] (1.5-.8-.7-.8-.2-.8-.2-1.5, -1.3) -- (example1Out);
  \draw [arrow2, dashed] (branch41) --  (example2Out);
  \draw [arrow2, dashed] (branch32) --  (example3Out);
  \draw [arrow2, dashed] (branch44) --  (example4Out);
  
 \end{tikzpicture}

%% file: images/motivation/features.tex
\centering
\scalebox{0.7}{\input{images/motivation/features-op.tex}}
\vspace{2ex}

\scalebox{0.7}{\input{images/motivation/features-type.tex}}
\vspace{2ex}

\hspace{1em}
\scalebox{0.7}{\input{images/motivation/features-class.tex}}

%% file: images/motivation/features-op.tex
\begin{tikzpicture}[font=\scriptsize][
  grow = right,
  sibling distance = 40em,
  level distance = 5em,
  edge from parent/.style = {draw, -latex},
  every node/.style = {font=\normalsize},
  sloped
]
\tikzstyle{root}=[rectangle,draw, align=center]

  \node [root, label=above:{\normalsize{$v \Coloneqq x \mid L \mid R$}}] at (1, 0) (op) {\normalsize{$\OPOP(v)$}};

  \node [circle, draw, dashed] at ($(op) + (-3.0, -3em)$) (opand) {};
  \node [circle, draw, dashed] at ($(op) + (-2.0, -3em)$) (opor) {};
  \node [circle, draw, dashed] at ($(op) + (-1.0, -3em)$) (opnot) {};
  \node [circle, draw, dashed] at ($(op) + (1.0, -3em)$) (opxor) {};
  \node [circle, draw, dashed] at ($(op) + (2.0, -3em)$) (opmul) {};
  \node [circle, draw, dashed] at ($(op) + (3.0, -3em)$) (opleaf) {};

  \draw [->] (op) -| node [near end, label=left:{\normalsize{$\ANDOP$}}, xshift=2mm] {} (opand);
  \draw [->] (op) -| node [near end, label=left:{\normalsize{$\OROP$}}, xshift=2mm] {} (opor);
  \draw [->] (op) -| node [near end, label=left:{\normalsize{$\NOTOP$}}, xshift=2mm] {} (opnot);
  \draw [->] (op) -| node [near end, label=left:{\normalsize{$\XOROP$}}, xshift=2mm] {} (opxor);
  \draw [->] (op) -| node [near end, label=left:{\normalsize{$\MULOP$}}, xshift=2mm] {} (opmul);
  \draw [->] (op) -| node [near end, label=left:{\normalsize{$\LEAFOP$}}, xshift=2mm] {} (opleaf);

\end{tikzpicture}

%% file: images/motivation/features-type.tex
\begin{tikzpicture}[font=\scriptsize][
  grow = right,
  sibling distance = 40em,
  level distance = 5em,
  edge from parent/.style = {draw, -latex},
  every node/.style = {font=\normalsize},
  sloped
]
\tikzstyle{root}=[rectangle,draw, align=center]

  \node [root, label=above:{\normalsize{$v \Coloneqq L \mid R$}}] at (1, 0) (type) {\normalsize{$\TYPE(v)$}};

  \node [circle, draw, dashed] at ($(type) + (-3.0, -3em)$) (rud) {};
  \node [circle, draw, dashed] at ($(type) + (-2.0, -3em)$) (sid) {};
  \node [circle, draw, dashed] at ($(type) + (-1.0, -3em)$) (ukd) {};
  \node [circle, draw, dashed] at ($(type) + (1.0, -3em)$) (inrand) {};
  \node [circle, draw, dashed] at ($(type) + (2.0, -3em)$) (inpub) {};
  \node [circle, draw, dashed] at ($(type) + (3.0, -3em)$) (inkey) {};

  \draw [->] (type) -| node [near end, label=left:{\normalsize{$\RUD$}}, xshift=2mm] {} (rud);
  \draw [->] (type) -| node [near end, label=left:{\normalsize{$\SID$}}, xshift=2mm] {} (sid);
  \draw [->] (type) -| node [near end, label=left:{\normalsize{$\UKD$}}, xshift=2mm] {} (ukd);
  \draw [->] (type) -| node [near end, label=left:{\small{$\INRANDOP$}}, xshift=2mm] {} (inrand);
  \draw [->] (type) -| node [near end, label=left:{\small{$\INPUBOP$}}, xshift=2mm] {} (inpub);
  \draw [->] (type) -| node [near end, label=left:{\small{$\INKEYOP$}}, xshift=2mm] {} (inkey);

\end{tikzpicture}

%% file: images/motivation/features-class.tex
\begin{tikzpicture}[font=\scriptsize][
  grow = right,
  sibling distance = 40em,
  level distance = 5em,
  edge from parent/.style = {draw, -latex},
  every node/.style = {font=\normalsize},
  sloped
]

 \tikzstyle{leaf}=[rectangle, draw, align=center, fill={rgb:black,1;white,2}]

  \node [circle, draw, dashed] at (0, 0) (prud) {};
  \node [leaf] at ($(prud) + (0, -3em)$) (rud) {\normalsize{$\RUD(x)$}};
  \draw [->] (prud) edge (rud);

  \node [circle, draw, dashed] at ($(prud) + (1.5, 0)$) (psid) {};
  \node [leaf] at ($(psid) + (0, -3em)$) (sid) {\normalsize{$\SID(x)$}};
  \draw [->] (psid) edge (sid);

  \node [circle, draw, dashed] at ($(psid) + (1.5, 0)$) (pukd) {};
  \node [leaf] at ($(pukd) + (0, -3em)$) (ukd) {\normalsize{$\UKD(x)$}};
  \draw [->] (pukd) edge (ukd);

\end{tikzpicture}

%% file: images/motivation/tree1.tex
\begin{tikzpicture}[font=\footnotesize]

 \tikzstyle{arrow1}=[thick,->,>=stealth,darkmain]
 \tikzstyle{arrow2}=[thick,->,>=stealth,black]
 \tikzstyle{root}=[rectangle,draw, align=center, thick]
 \tikzstyle{leaf}=[rectangle, draw, align=center, fill={rgb:black,1;white,2}]

 \node at (0,0) [root] (root0) {\normalsize{$\OPOP$(x)}};
 
 \node at (-.8, -.6) [] (rootConditionL) {\color{black} $\ANDOROP$(x)};
 \node at (.8, -.6) [] (rootConditionR) {\color{red} $\XOROP$(x)};
 
  \draw (root0) -- (rootConditionL);
 \draw[red] (root0) -- (rootConditionR);
 
 \node at (-1.5, -1.3) [leaf] (branch11) {\small{$\SID$}};
 \node at (1.5, -1.3) [root] (branch12) {\small{$\OPOP$(R)}};
 
  \draw [arrow2] (rootConditionL) -- (branch11);
 \draw [arrow2, red] (rootConditionR) -- (branch12);
 
 \node at (1.5-1.3, -1.3-.6) [] (opCondition1)  {\color{black} $\LEAFOP$(R) };
 \node at (1.5, -1.3-.6) [] (opCondition2)  {\color{black} $\ANDOROP$(R)  };
 \node at (1.5+1.2, -1.3-.6) [] (opCondition3)  {\color{red} $\XOROP$(R) };
 
 \draw (branch12)  -- (opCondition1);
 \draw (branch12)  -- (opCondition2);
  \draw[red] (branch12)  -- (opCondition3);
 
 \node at (1.5-1.3-.7, -1.3-.6-.7) [leaf] (branch21) {\small{$\RUD$}};
 \node at (1.5, -1.3-.6-.7) [root] (branch22) {\small{$\TYPE$(L)}};
 \node at (1.5+1.2+.7, -1.3-.6-.7) [root] (branch23) {\small{$\TYPE$(L)}};
 
 \draw[arrow2] (opCondition1) -- (branch21);
 \draw[arrow2] (opCondition2) -- (branch22);
 \draw[arrow2, red] (opCondition3) -- (branch23);
 
   \node at (1.5-1.3, -1.3-.6-.7-.6) [] (typeL1condL) {\color{black} $\SID$(L)};
 \node at (1.5-.3, -1.3-.6-.7-.6) [] (typeL1condR) {\color{black}$\RUD$(L)};
 
  \draw (branch22) -- (typeL1condL);
 \draw (branch22) -- (typeL1condR);
 
 \node at (1.5-1.3-.2-.1, -1.3-.6-.7-.6-.7)  [leaf] (branch31) {\small{$\SID$}};
 \node at (1.5-.3-.2, -1.3-.6-.7-.6-.7)  [leaf] (branch32) { \small{$\RUD$}};
 
  \draw[arrow2] (typeL1condL) -- (branch31);
  \draw[arrow2] (typeL1condR) -- (branch32);
 
   \node at (1.5+1.2+.7-.8, -1.3-.6-.7-.6) [] (typeL2condL) {\color{black} $\SID$(L)};
 \node at (1.5+1.2+.7+.8, -1.3-.6-.7-.6) [] (typeL2condR) {\color{red}$\RUD$(L)};
 
  \draw (branch23) -- (typeL2condL);
 \draw[red] (branch23) -- (typeL2condR);
  
  \node at (1.5+1.2+.7-.8-.1, -1.3-.6-.7-.6-.7) [leaf] (branch33) {\small{$\RUD$}};
  \node at (1.5+1.2+.7+.8, -1.3-.6-.7-.6-.7) [root] (branch34) {\color{red} \small{$\RUD$?$\UKD$}};
  
  \draw[arrow2] (typeL2condL) -- (branch33);
  \draw[arrow2, red] (typeL2condR) -- (branch34);
 
  \node at (-1.5-.7, -1.3-.7) [root, minimum width=14mm, minimum height=5mm] (example1) {$\{ \mathlst{n4}, \mathlst{n7}, \mathlst{n8}, \mathlst{n9}\}$};
  \node at (-1.5-.7, -1.3-.7) [root, minimum width=24mm, minimum height=7mm] (example1Out) {$\{ \mathlst{n4}, \mathlst{n7}, \mathlst{n8}, \mathlst{n9}\}$};
    \draw[dashed] (branch11.west) -- (-1.5-.7, -1.3);
   \draw[arrow2, dashed] (-1.5-.7, -1.3) -- (example1Out);
   
   \node at (1.5-1.3-.2-.1, -1.3-.6-.7-2.6) [root, minimum width=7mm, minimum height=5mm] (example2) {$\{ \mathlst{n5}\}$}; 
   \node at (1.5-1.3-.2-.1, -1.3-.6-.7-2.6) [root, minimum width=10mm, minimum height=7mm] (example2Out) {$\{ \mathlst{n5}\}$};
   \draw[arrow2, dashed] (branch31.south) -- (example2Out.north);
   
   \node at (1.5-.3-.2, -1.3-.6-.7-2.6) [root, minimum width=7mm, minimum height=5mm] (example3) {$\{ \mathlst{n6}\}$}; 
   \node at (1.5-.3-.2, -1.3-.6-.7-2.6) [root, minimum width=10mm, minimum height=7mm] (example3Out) {$\{ \mathlst{n6}\}$}; 
   \draw[arrow2, dashed] (branch32.south) -- (example3Out.north);
   
    \node at (1.5+1.2+.7-.8-.1, -1.3-.6-.7-2.6) [root, minimum width=14mm, minimum height=5mm] (example4) {$\{ \mathlst{n2}, \mathlst{n3}\}$}; 
   \node at (1.5+1.2+.7-.8-.1, -1.3-.6-.7-2.6) [root, minimum width=16mm, minimum height=7mm] (example4Out) {$\{ \mathlst{n2}, \mathlst{n3}\}$}; \draw[arrow2, dashed] (branch33.south) -- (example4Out.north);
   
   \node at (1.5+1.2+.7+.8, -1.3-.6-.7-2.6) [root, minimum width=14mm, minimum height=5mm] (example5) {\color{red} $\{ \mathlst{b4}, \mathlst{n1}\}$}; 
   \node at (1.5+1.2+.7+.8, -1.3-.6-.7-2.6) [root, minimum width=16mm, minimum height=7mm] (example5Out) {\color{red} $\{ \mathlst{b4}, \mathlst{n1}\}$};
   \draw[arrow2, dashed, red] (branch34.south) -- (example5Out.north);

   \node at (-1.5-.4, -1.3-.6-.7-2.6)  [root, minimum width=14mm, minimum height=5mm] (example6) {$\{ \mathlst{b1}, \mathlst{b2}, \mathlst{b3}\}$}; 
   \node at (-1.5-.4, -1.3-.6-.7-2.6)  [root, minimum width=20mm, minimum height=7mm] (example6Out) {$\{ \mathlst{b1}, \mathlst{b2}, \mathlst{b3}\}$};
   \draw[dashed] (branch21.west) -- (-1.5-.4, -1.3-.6-.7);
   \draw[arrow2, dashed] (-1.5-.4, -1.3-.6-.7) -- (example6Out);
\end{tikzpicture}

%% file: images/motivation/tree2.tex
\begin{tikzpicture}[font=\footnotesize]

 \tikzstyle{arrow1}=[thick,->,>=stealth,darkmain]
 \tikzstyle{arrow2}=[thick,->,>=stealth,black]
 \tikzstyle{root}=[rectangle,draw, align=center, thick]
 \tikzstyle{rootGreen}=[rectangle,draw, align=center, color=cadmiumgreen, very thick]
 \tikzstyle{leaf}=[rectangle, draw, align=center, fill={rgb:black,1;white,2}]
 \tikzstyle{leafCE}=[rectangle, draw, align=center, color=red, very thick, fill={rgb:black,1;white,2}] \tikzstyle{leafCorrect}=[rectangle, draw, align=center, color=green, very thick, fill={rgb:black,1;white,2}]
 
 \tikzstyle{dashed_rectangle}=[rectangle,draw, align=center, dashed, thick]

 \node at (0,0) [root] (root0) {\normalsize{$\OPOP$(x)}};
 
 \node at (-.8, -.6) [] (rootConditionL) {\color{black} $\ANDOROP$(x)};
 \node at (.8, -.6) [] (rootConditionR) {\color{blue} $\XOROP$(x)};
 
  \draw (root0) -- (rootConditionL);
 \draw[blue] (root0) -- (rootConditionR);
 
 \node at (-1.5, -1.3) [leafCE]  (branch11) {\small{$\SID$}};
 \node at (1.5, -1.3) [root] (branch12) {\small{$\OPOP$(R)}};
 
  \draw [arrow2] (rootConditionL) -- (branch11);
 \draw [arrow2, blue] (rootConditionR) -- (branch12);
 
 \node at (1.5-1.3, -1.3-.6) [] (opCondition1)  {\color{black} $\LEAFOP$(R) };
 \node at (1.5, -1.3-.6) [] (opCondition2)  {\color{black} $\ANDOROP$(R)  };
 \node at (1.5+1.2, -1.3-.6) [] (opCondition3)  {\color{blue} $\XOROP$(R) };
 
 \draw (branch12)  -- (opCondition1);
 \draw (branch12)  -- (opCondition2);
  \draw[blue] (branch12)  -- (opCondition3);
 
 \node at (1.5-1.3-.7, -1.3-.6-.7) [leafCE] (branch21) {\small{$\RUD$}};
 \node at (1.5, -1.3-.6-.7) [root] (branch22) {\small{$\TYPE$(L)}};
 \node at (1.5+1.2+.7, -1.3-.6-.7) [root] (branch23) {\small{$\TYPE$(L)}};
 
 \draw[arrow2] (opCondition1) -- (branch21);
 \draw[arrow2] (opCondition2) -- (branch22);
 \draw[arrow2, blue] (opCondition3) -- (branch23);
 
   \node at (1.5-1.3, -1.3-.6-.7-.6) [] (typeL1condL) {\color{black} $\SID$(L)};
 \node at (1.5-.3, -1.3-.6-.7-.6) [] (typeL1condR) {\color{black}$\RUD$(L)};
 
  \draw (branch22) -- (typeL1condL);
 \draw (branch22) -- (typeL1condR);
 
 \node at (1.5-1.3-.2-.1, -1.3-.6-.7-.6-.7)  [leafCE] (branch31) {\small{$\SID$}};
 \node at (1.5-.3-.2, -1.3-.6-.7-.6-.7)  [leafCE] (branch32) { \small{$\RUD$}};
 
  \draw[arrow2] (typeL1condL) -- (branch31);
  \draw[arrow2] (typeL1condR) -- (branch32);
 
   \node at (1.5+1.2+.7-.8, -1.3-.6-.7-.6) [] (typeL2condL) {\color{black} $\SID$(L)};
 \node at (1.5+1.2+.7+.8, -1.3-.6-.7-.6) [] (typeL2condR) {\color{blue}$\RUD$(L)};
 
  \draw (branch23) -- (typeL2condL);
 \draw[blue] (branch23) -- (typeL2condR);
  
  \node at (1.5+1.2+.7-.8-.1, -1.3-.6-.7-.6-.7) [leafCE] (branch33) {\small{$\RUD$}};
  \node at (1.5+1.2+.7+.8, -1.3-.6-.7-.6-.7)  [rootGreen]  (branch34) { \small{f(x)}};
  
  \draw[arrow2] (typeL2condL) -- (branch33);
  \draw[arrow2] (typeL2condR) -- (branch34);
 
  \node at (1.5+1.2+.7+.8-.3, -1.3-.6-.7-.6-.8-.5) [] (fCondL) {\color{black} true};
   \node at (1.5+1.2+.7+.8+.5, -1.3-.6-.7-.6-.8-.5) [] (fCondR) {\color{black} false}; 
   
   \draw[blue] (branch34) -- (fCondL);
   \draw[blue] (branch34) -- (fCondR);
   
   \node at (1.5+1.2+.7+.8-.5+.1, -1.3-.6-.7-.6-.8-.6-.6) [leafCorrect] (branch41) {\small{$\UKD$}};
   \node at (1.5+1.2+.7+.8-.5+1.1, -1.3-.6-.7-.6-.8-.6-.6) [leafCorrect] (branch42) {\small{$\RUD$}};
  
   \draw[arrow2, blue] (fCondL) -- (branch41);
   \draw[arrow2, blue] (fCondR) -- (branch42);
    \node[dashed_rectangle, minimum width=5.6em, minimum height=6.2em] at ([yshift=-.5em, xshift=.1em]fCondL.north east) {};

\end{tikzpicture}

%% file: images/motivation/tree3.tex
\begin{tikzpicture}[font=\footnotesize]

\tikzstyle{arrow1}=[thick,->,>=stealth,darkmain]
 \tikzstyle{arrow2}=[thick,->,>=stealth,black]
 \tikzstyle{root}=[rectangle,draw, align=center,thick]
 \tikzstyle{rootGreen}=[rectangle,draw, align=center, color=cadmiumgreen, very thick]
 \tikzstyle{leaf}=[rectangle, draw, align=center, fill={rgb:black,1;white,2}]
 \tikzstyle{leafCE}=[rectangle, draw, align=center, color=red, very thick, fill={rgb:black,1;white,2}] \tikzstyle{leafCorrect}=[rectangle, draw, align=center, color=green, very thick, fill={rgb:black,1;white,2}]
 
 \node at (0,0) [root] (root0) {\normalsize{$\TYPE$(L)}};
 
  \node at (-1.2, -.6) [] (typeCondition1)  {\color{black} $\RUD$(L) };
 \node at (0, -.6) [] (typeCondition2)  {\color{black} $\UKD$(L) };
 \node at (1.1, -.6) [] (typeCondition3)  {\color{black} $\SID$(L) };
 
 \draw (root0) -- (typeCondition1);
 \draw (root0) -- (typeCondition2);
 \draw (root0) -- (typeCondition3);
 
 \node at (-1.2-.7, -.6-.7) [root] (branch11) {\small{$\OPOP$(x)}};
  \node at (0, -.6-.7) [leafCE] (branch12) {\small{$\RUD$}};
\node at (1.1+.7, -.6-.7) [root] (branch13) {\small{$\OPOP$(R)}};

\draw[arrow2] (typeCondition1) -- (branch11);
\draw[arrow2] (typeCondition2) -- (branch12);
\draw[arrow2] (typeCondition3) -- (branch13);

\node at (-1.2-.7-.8, -.6-.7-.6) [] (opXL) {\color{black} $\XOROP$(x)};
\node at (-1.2-.7+.8, -.6-.7-.6) [] (opXR) {\color{black} $\ANDOROP$(x)};

\draw (branch11) -- (opXL);
\draw (branch11) -- (opXR);

\node at (-1.2-.7-.8-.4, -.6-.7-.6-.7) [root] (branch21) {\small{f(x)}};
\node at (-1.2-.7+.8+.2, -.6-.7-.6-.7) [leafCE] (branch22) {\small{$\SID$}};

\draw[arrow2] (opXL) -- (branch21);
\draw[arrow2] (opXR) -- (branch22);

\node at (1.1+.7-.8, -.6-.7-.6) [] (opRL) {\color{black} $\ANDOROP$(R)};
\node at (1.1+.7+.8, -.6-.7-.6) [] (opRR) {\color{black} $\XOROP$(R)};

\draw (branch13) -- (opRL);
\draw (branch13) -- (opRR);

\node at (1.1+.7-.8-.4, -.6-.7-.6-.7) [leafCE] (branch23) {\small{$\SID$}};
\node at (1.1+.7+.8+.2, -.6-.7-.6-.7) [leafCE] (branch24) {\small{$\RUD$}};

\draw[arrow2] (opRL) -- (branch23);
\draw[arrow2] (opRR) -- (branch24);

\node at (-1.2-.7-.8-.4-.8, -.6-.7-.6-.7-.6) [] (fxL) {\color{black} true};
\node at (-1.2-.7-.8-.4+.8, -.6-.7-.6-.7-.6) [] (fxR) {\color{black} false};

\draw (branch21) -- (fxL);
\draw (branch21) -- (fxR);

\node at (-1.2-.7-.8-.4-.8-.4, -.6-.7-.6-.7-.6-.7) [leafCorrect] (branch31) {\small{$\UKD$}};
\node at (-1.2-.7-.8-.4+.8+.2, -.6-.7-.6-.7-.6-.7) [leafCorrect] (branch32) {\small{$\RUD$}};

\draw[arrow2] (fxL) -- (branch31);
\draw[arrow2](fxR) -- (branch32);
\end{tikzpicture}

%% file: images/learner/dtl.tex
\begin{algorithmic}[1]
\Input Examples, $\mathcal{E} = \{ (x_1, \TYPE(x_1)), \dots, (x_n, \TYPE(x_n)) \}$
\Input Pre-defined features, $\mathcal{F} = \{ f_1, f_2, \dots, f_k \}$
\Output Classifier $\mathcal{T}$ which is consistent with provided examples

\If {all examples $(x, \TYPE(x)) \in \mathcal{E}$ have the same label $\TYPE(x) = t$}
  \State \textbf{return} $\mathcal{T} = \operatorname{LeafNode}(t)$
\EndIf

  \If {$\not \exists f \in \mathcal{F}$ such that $H(\mathcal{E} \mid f) < H(\mathcal{E})$}
    \State $\mathcal{F} \coloneqq \mathcal{F} \union \FeatureSyn(\mathcal{E})$
  \EndIf

  \State $\mathcal{T} = \operatorname{DecisionNode}(f^*)$, where $f^* = \argmin_{f \in \mathcal{F}} H(\mathcal{E} \mid f)$
  \For {valuation $i$ of feature $f*$}
    \State $\mathcal{T}_i = \DTL(\mathcal{E} \rvert_{f^*(x) = i}, \mathcal{F} \setminus \{ f^* \})$
    \State Add edge from $\mathcal{T}$ to $\mathcal{T}_i$ with label $f^*(x) = i$
  \EndFor

\State \textbf{return} $\mathcal{T}$
\end{algorithmic}


%% file: images/learner/feature-syn.tex
\begin{algorithmic}[1]
\Input Examples, $\mathcal{E} = \{ (x_1, \TYPE(x_1)), \dots, (x_n, \TYPE(x_n)) \}$
\Output Feature $f$ with positive information gain, or $\bot$ to indicate failure

\State Let $\mathcal{S}$ be the meta-rules defined in Figure~\ref{fig:learner:sygus:metarules}, i.e. the
  \emph{hypothesis space}

\ForEach {relation schema $r$ defined in $\mathcal{S}$}
  \ForEach {subset $S$ of meta-rules corresponding to the schema}
    \ForEach {choice $p_{in}$, $q_{in}$, and nested relational predicates}
      \State Let $f$ be the corresponding instantiation of the meta-rules in $S$
      \If {$h(\mathcal{E} \mid f) \lneq h(\mathcal{E})$}
        \State \textbf{return} $f$
      \EndIf
    \EndFor
  \EndFor
\EndFor

\State \textbf{return} $\bot$
\end{algorithmic}


%% file: images/prover/AST_proof.tex

\begin{tikzpicture}[font=\scriptsize] 
  [
    grow                    = right,
    sibling distance        = 8em,
    level distance          = 2em,
    edge from parent/.style = {draw, -latex},
    every node/.style       = {font=\footnotesize},
    sloped
  ]
  \node [root] at (-1.8, 0) {$\vee$}  
  child{ node [root, label=left:{\textcolor{darkermain}{$g_1(L, k_1)$}}] at (-.1, .7) {$\vee$}  
   child{node [env] at (-.1, .7)  {$k_1$} 
   edge from parent node  [below, xshift=1mm, align=center]{$\LCOP$}
   }
   child{
   node [env] at (.1, .7)  {$r1$}
   edge from parent node  [below, xshift=-1mm, align=center]{$\RCOP$}
   }
  edge from parent node  [below, align=center]{}
  }
  child{ node [root, label=right:{\textcolor{darkermain}{$g_2(R, k_2)$}}] at (.1, .7) {$\lnot$}  
   child{ node [env] at (.8, .7)  {$k_2$}
   edge from parent node  [below, xshift=-1mm, align=center]{$\LCOP$}
   }
   edge from parent node  [below, align=center]{}
  };
  
 \node[minimum width=28mm, minimum height=14mm, dashed, draw=black, thick, align=center] (PR1) at (-2.9, -1.2) {};
  \node[minimum width=22mm, minimum height=14mm, dashed, draw=black, thick, align=center, label=right:{$k_1 = k_2$}] (PR1) at (-.2, -1.2) {};
  \node at (-1,0) {$\SID(x)$};
\end{tikzpicture}